\documentclass[
submission
 , nomarks
]{dmtcs-episciences}

\usepackage[utf8]{inputenc}
\usepackage{subfigure}

\usepackage[round]{natbib}

\newcommand{\inSet}{T_{v,C}^{\operatorname{in}}}
\newcommand{\outSet}{T_{v,C}^{\operatorname{out}}}
\newcommand{\ie}{i.\,e.\ }
\newcommand{\eg}{e.\,g.\ }
\newcommand{\dist}{\operatorname{dist}}
\newcommand{\abs}[1]{\vert #1 \vert}

\usepackage{amsfonts}
\usepackage{amssymb}

\usepackage{amsmath}
\usepackage{constants}
\usepackage{amsthm}
\usepackage{tikz}      
\usetikzlibrary{automata, arrows,snakes}
\newtheorem{theorem}{Theorem}[section]

\newtheorem{lemma}[theorem]{Lemma}
\newtheorem{corollary}[theorem]{Corollary}
\newtheorem{claim}[theorem]{Claim}

\newtheorem{definition}[theorem]{Definition}

\newtheorem{remark}[theorem]{Remark}
\newtheorem{note}[theorem]{Note}

\author{Isolde Adler\affiliationmark{1}
  \and Noleen Köhler\affiliationmark{1}}
\title[On graphs that are far from being Hamiltonian]{An explicit construction of graphs of bounded degree that are far from being Hamiltonian}
\affiliation{
  School of Computing, University of Leeds, UK}
\received{2021-01-20}

\revised{2021-09-15,2022-01-01}

\accepted{2022-01-02}

\begin{document}
\publicationdetails{24}{2022}{1}{3}{7109}
\maketitle
\begin{abstract}
  Hamiltonian cycles in graphs were first studied in the 1850s. 
  Since then, an impressive amount of research has been dedicated to identifying classes 
  of graphs that allow Hamiltonian cycles, and to related questions.
  The corresponding decision problem, that asks whether a given graph is Hamiltonian (\ie admits a Hamiltonian cycle), is one of Karp's famous NP-complete problems. 
  
  In this paper we study graphs of bounded degree that are \emph{far} from being Hamiltonian,
  where a graph
  $G$ on $n$ vertices is \emph{far} from being Hamiltonian, 
  if modifying a constant fraction of $n$ edges 
  is necessary
  to make $G$ Hamiltonian. We give an explicit deterministic construction of a class of graphs of bounded degree that are locally Hamiltonian,
  but (globally) far from being Hamiltonian. Here, \emph{locally Hamiltonian} means that  
  every subgraph induced by the neighbourhood of a small vertex set
  appears in some Hamiltonian graph. More precisely, we obtain graphs which differ in $\Theta(n)$ edges from any Hamiltonian graph,
  but non-Hamiltonicity cannot be detected in the neighbourhood of $o(n)$ vertices. 
  
  Our  class of graphs yields a class of hard instances for one-sided error property testers with linear query complexity. 
  It is known that any property tester (even with two-sided error) requires   a linear number of queries to test Hamiltonicity (Yoshida, Ito, 2010). This is proved via a randomised construction of hard instances. In contrast, our construction is deterministic. So far only very few deterministic constructions of hard instances for property testing are known. We believe that our construction may lead to future insights in graph theory and towards a characterisation of the  properties that are testable in the bounded-degree	model.

\end{abstract}

\section{Introduction}
A \emph{Hamiltonian cycle} in a graph $G$ is a cycle that visits every vertex of $G$ exactly once.
A graph $G$ is \emph{Hamiltonian} if $G$ contains a Hamiltonian cycle.
Research on Hamiltonian graphs has a long and rich history, see \eg~\hspace{-4pt}\cite{Gould91}.
Dirac's early Theorem~\cite{Dirac52} gave sufficient conditions for Hamiltonicity, and
subsequently, many further classes of Hamiltonian graphs were identified. Interestingly, it was shown by Robinson and Wormald that for $d\geq 3$, almost all $d$-regular graphs are Hamiltonian~\cite{Robinson1994AlmostAR}. 

Hamiltonian graphs play an important role in routing, including network design~\cite{ChuWY05,Parhami05},
circuit design~\cite{WangCC12}, and computer graphics~\cite{ZhangZH13}, 
as well as in scheduling via tight links to the Travelling Salesperson Problem.
Deciding whether a given graph is Hamiltonian is NP-complete~\cite{Karp72}, even on cubic planar 
graphs~\cite{GareyJT76}.

In this paper we study graphs of bounded degree that are far from being Hamiltonian, 
where intuitively, a graph $G$ is \emph{far} from being Hamiltonian if many 
edge modifications (insertions or deletions) are necessary to make $G$ Hamiltonian (note that deletions may
help, because of the degree bound).

\textbf{Motivation.} The wider motivation for our study stems from the well-known tight connection between \emph{structural} properties of graphs and their \emph{algorithmic} properties, which has been used successfully for designing efficient algorithms for numerous problems, all the way to reaching the boundaries of efficient solvability. 
Hence for many important graph properties (where by \emph{property} we simply mean an isomorphism closed graph class), the structure of graphs having the property is studied in great detail. We propose studying the structure of graphs that are far from having a given property.
This is motivated by the area of property testing, in which computational decision problems are relaxed to distinguishing graphs that have a certain property from graphs that are far from having the property.  We study graphs that are far from being Hamiltonian. Hamiltonicity is known to be hard for property testing \cite{DBLP:journals/ieicet/YoshidaI10,Goldreich20}. However this is shown via a randomised construction. We give an explicit, deterministic construction of hard instances for testing Hamiltonicity. In computer science, explicit constructions are often of interest as they can further our understanding of the complexity of related computational  problems, and knowledge of explicit structural  properties and parameters can be exploited.

We now give more details. For a given $\epsilon$ in the real interval $[0,1]$, we say that a graph $G$ of maximum degree $d$ with $n$ vertices is \emph{$\epsilon$-close} to being Hamiltonian, if at most $\epsilon d n$ edge modifications (insertions or deletions) are needed to make $G$ Hamiltonian, and $G$ is \emph{$\epsilon$-far} from being Hamiltonian otherwise. Note that $dn$ is an upper bound on the total number of edges in an $n$-vertex graph of degree at most $d$. 

\begin{figure}[htbp]
	\begin{center}
		\subfigure[Caterpillar]{
			\begin{tikzpicture}[>=stealth']
			\definecolor{C1}{RGB}{1,1,1}
			\definecolor{C2}{RGB}{0,204,204}
			\definecolor{C3}{RGB}{204,0,0}
			\definecolor{C4}{RGB}{0,204,0}
			
			\tikzstyle{ns1}=[line width=1]
			
			\def \breadthDist {0.6}
			\def \heightDist {0.6}
			\def \numberOfVertices{10}
			\def \oneLessVertex{9}
			\foreach \x in {1,...,\numberOfVertices} {
				\node[draw,circle,fill=black,inner sep=0pt, minimum width=4pt] (\x) at (\x*\breadthDist,0) {};
				\node[draw,circle,fill=black,inner sep=0pt, minimum width=4pt] (\x1) at (\x*\breadthDist,-\heightDist) {};
			}
			
			\foreach \x in {1,...,\oneLessVertex} {
				\pgfmathtruncatemacro{\xx}{\x+1};
				\path[C1,ns1]          (\x)  edge  (\x1);
				\path[C1,ns1]          (\x)  edge  (\xx);}
			\path[C1,ns1]          (\numberOfVertices*\breadthDist-2*\breadthDist,0)  edge  (\numberOfVertices*\breadthDist-2*\breadthDist,-\heightDist);
			\path[C1,ns1]          (\numberOfVertices*\breadthDist,0)  edge  (\numberOfVertices*\breadthDist,-\heightDist);

			\node at (2,-2) {};

			\end{tikzpicture}}
		\hfil
		\subfigure[$C_4$'s arranged in a cycle.]{\begin{tikzpicture}[>=stealth']
			\definecolor{C1}{RGB}{1,1,1}
			\definecolor{C2}{RGB}{0,204,204}
			\definecolor{C3}{RGB}{204,0,0}
			\definecolor{C4}{RGB}{0,204,0}
			
			\tikzstyle{ns1}=[line width=1]
			\def \numberOfCFours{11}
			\def \oneLessCFours{10}
			\def \radiusCFour {0.3}
			\def \radius {1.8}
			\foreach \x in {1,...,\numberOfCFours} {
				\pgfmathtruncatemacro{\degree}{360*\x/\numberOfCFours};
				\foreach \y in {1,...,4}{
					\pgfmathtruncatemacro{\littleDegree}{\y*90+\degree};
					\pgfmathtruncatemacro{\name}{\x*4+\y};
					\node[draw,circle,fill=black,inner sep=0pt, minimum width=4pt] [shift={(\degree:\radius)}](\name) at (\littleDegree:\radiusCFour) {};
			}}
			\foreach \x in {1,...,\numberOfCFours} {
				\pgfmathtruncatemacro{\firstNode}{\x*4+4};
				\pgfmathtruncatemacro{\secondNode}{\x*4+1};
				\path[C1,ns1]          (\firstNode)  edge  (\secondNode);
				\foreach \y in {1,2,3}{
					\pgfmathtruncatemacro{\firstNode}{\x*4+\y};
					\pgfmathtruncatemacro{\secondNode}{\firstNode+1};
					\path[C1,ns1]          (\firstNode)  edge  (\secondNode);
			}}
			\foreach \x in {1,...,\oneLessCFours} {
				\draw [C1,ns1,domain={(360*\x/\numberOfCFours)+atan(\radiusCFour/\radius)}:{(360*(\x+1)/\numberOfCFours)-atan(\radiusCFour/\radius)}] plot ({cos(\x)*sqrt(\radius^2+\radiusCFour^2)}, {sin(\x)*sqrt(\radius^2+\radiusCFour^2)});}
			\draw [C1,ns1,domain={atan(\radiusCFour/\radius)}:{(360/\numberOfCFours)-atan(\radiusCFour/\radius)}] plot ({cos(\x)*sqrt(\radius^2+\radiusCFour^2)}, {sin(\x)*sqrt(\radius^2+\radiusCFour^2)});	
			\end{tikzpicture}
			}
		\caption{Example graphs which are far from being Hamiltonian but are not locally Hamiltonian.} \label{fig:examplesIntroduction}
	\end{center}
\end{figure}
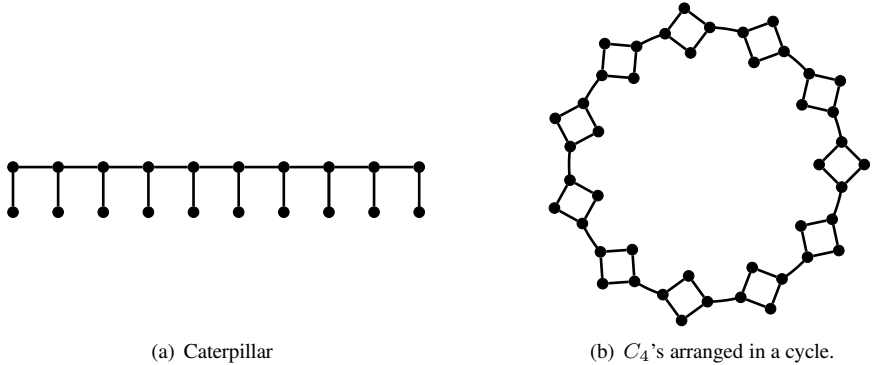

It is easy to find graphs that are far from being Hamiltonian. For example, let $G$ be a caterpillar graph on $n=2k$ vertices as shown 
in Figure~\ref{fig:examplesIntroduction} for $k=10$ (\ie $G$ is a path of length $k-1$ where every vertex has a pendant edge). With a degree bound of at most $3$, $G$ 
is $1/13$-far
from being Hamiltonian, because $n/4$ edges need to be added to make $G$ $2$-connected. 
As another example, consider the graph $H$ consisting of $k$ $4$-cycles ($C_4$'s)
arranged in a cycle as shown in Figure~\ref{fig:examplesIntroduction} for $k=11$. Assume $k>1$. The graph $H$ has $n=4k$ vertices and, with a degree bound of $3$, $H$ is $1/25$-far
from being Hamiltonian. This is because any Hamiltonian cycle in a graph has to traverse both edges incident to any vertex of degree $2$. Hence in $H$ a Hamiltonian cycle would have to traverse all four edges of every $C_4$. To avoid this we have to increase the degree of at least one of the degree $2$ vertices for every $C_4$ and hence we have to add at least $n/8$ edges to make $H$ Hamiltonian.

In both examples it is possible to see locally, in
the neighbourhood of a constant number of vertices, that the graphs are not Hamiltonian.
We ask whether there exist graphs that locally look as if they might be Hamiltonian, but globally they are far from being Hamiltonian, and we give a positive answer to this. More precisely, for $\delta\in (0,1]$ we define a graph $G$ with vertex set $V$ and $|V|=n$ to be \emph{$\delta$-locally Hamiltonian} if the subgraph induced by 
the neighbourhood of any subset 
$S\subseteq V$  with $|S|\leq \delta n$
appears in some Hamiltonian graph on $n$ vertices. We then show the following by giving an explicit construction (cf.~Theorem~\ref{thm:existenceOfLocallyButFarFromHamiltonianGraphs}).\smallskip

\textit{There is a $d\in \mathbb{N}$ and there are constants $\delta:=\delta(d),\epsilon:=\epsilon(d) \in (0,1)$ and a sequence of $d$-bounded degree graphs $(G_N)_{N\in \mathbb{N}}$ of increasing order, such that $G_N$ is $\delta$-locally Hamiltonian and $\epsilon$-far from being Hamiltonian for every $N\in \mathbb{N}$}.\smallskip

A similar approach was taken in \cite{3colPropertyTesting} for  $3$-colourability, where graphs, which are far from being $3$-colourable but locally look $3$-colourable, are implicitly obtained using a reduction from the constraint satisfaction problem (CSP). An explicit construction of a CSP, which is far from being satisfiable but every sublinear subset of constraints is satisfiable, is given. To our knowledge this is the only other known deterministic construction of a similar kind.

\textbf{Property Testing.} Property testing on graphs is a framework for studying sampling-based algorithms that solve a relaxation of classical decision problems. Given a graph $G$ and a property $P$ (\eg triangle-freeness), the goal of a property testing algorithm, called a \emph{property tester}, is to distinguish if a graph satisfies $P$ or is \emph{far} from satisfying $P$, where the definition of \emph{far} depends on the model. Property testing of dense graphs is well understood through its tight links with Szemer\'edi's regularity Lemma~\cite{alon2009combinatorial}.
In~\cite{GoldreichRon2002}, Goldreich and Ron introduced property testing on bounded-degree graphs, and since then much 
attention has been paid to property testing in sparse graphs. Nevertheless, our understanding of testability of properties in such graphs is still limited. 
In the \emph{bounded-degree graph model}~\cite{GoldreichRon2002}, the tester has oracle access to the input graph $G$ with maximum degree $d$, where $d$ is constant, and is allowed to perform \emph{neighbour queries} to the oracle. That is, for any specified vertex $v$ and index $i\leq d$, the oracle returns the $i$-th neighbour of $v$ if it exists or a special symbol $\bot$ otherwise in constant time. 
A graph $G$ with $n$ vertices is called \emph{$\varepsilon$-far} from satisfying a property $P$, if one needs to modify more than $\varepsilon dn$ edges to make it satisfy $P$.
The goal now becomes to distinguish, with probability at least $2/3$, if $G$ satisfies a property $P$ or is $\varepsilon$-far from satisfying $P$, for any specified proximity parameter $\varepsilon\in (0,1]$. 
Here the choice of success probability $2/3$ is arbitrary, any constant strictly greater than $1/2$ can be used.
A property $P$ is \emph{testable with query complexity $q(n)$} in the bounded-degree model, if for every $\varepsilon\in (0,1]$ there is an algorithm (an \emph{$\varepsilon$-tester}),
that makes this distinction while using at most $q(n)$ oracle queries, where $n$ is the size of the
input graph. Property $P$ is testable with \emph{one-sided error} if
instances in $P$ are always correctly identified. If $q$ is independent of $n$, we have \emph{constant} query complexity. Here the constant can depend on $\varepsilon$ and $d$. 

So far, it is known that some properties are constant-query testable, including subgraph-freeness, $k$-edge connectivity, cycle-freeness, being Eulerian, degree-regularity~\cite{GoldreichRon2002}, minor-freeness~\cite{benjamini2010every,hassidim2009local,kumar2019random}, hyperfinite properties \cite{NewmanSohler2013}, $k$-vertex connectivity~\cite{yoshida2012property,forster2019computing}, and subdivision-freeness~\cite{kawarabayashi2013testing}. On the other hand there are some properties whose query complexity is sublinear but not constant, \eg bipartiteness \cite{GoldreichRon2002,DBLP:journals/combinatorica/GoldreichR99}. Furthermore, there are some properties for which no tester with sublinear query complexity exists, \eg Hamiltonicity \cite{DBLP:journals/ieicet/YoshidaI10,Goldreich20}, 3-colourability \cite{3colPropertyTesting}, independent set size \cite{Goldreich20}. Note that for any computable property there is a linear query complexity property tester, \ie the tester, which accesses the entire graph and then uses any exact algorithm for the property, as we do not bound the running time of a property tester. We further want to point out that it is not true that NP-hard problems are in general hard for property testing, as \cite{NewmanSohler2013} shows that any property is constant query testable on the class of bounded-degree, planar graphs and many problems  remain NP-hard even on bounded degree, planar graphs, \eg Hamiltonicity \cite{GareyJT76}.

The major open problem in the area of property testing is finding  a full characterisation of the testable properties in the bounded degree model. Ito et al.~\cite{ito2019characterization} gave characterisations of one-sided error constant-query testable monotone graph properties, and one-sided error testable hereditary graph properties in the bounded-degree (directed and undirected) graph model.
The characterisation is based on the presence of many forbidden \emph{configurations} -- \emph{subgraphs} in the case of monotone properties and \emph{induced subgraphs} in the case of hereditary properties. Note that Hamiltonicity is a property that is neither monotone nor hereditary. Hence we believe that our results advance our understanding of testability of such properties.

We show that  any one-sided error property tester with sublinear query complexity needs to accept any locally Hamiltonian graph. Since every graph in our constructed class is locally Hamiltonian and far from being Hamiltonian we get the following, previously known lower-bound (see \cite{DBLP:journals/ieicet/YoshidaI10,Goldreich20}) as a direct consequence of Theorem~\ref{thm:existenceOfLocallyButFarFromHamiltonianGraphs}  (cf. Corollary~\ref{cor:main}).\smallskip

\emph{Hamiltonicity is not testable with one-sided error and query complexity $o(n)$
	in the bounded-degree model.}\smallskip

This provides evidence that using deterministic constructions is a viable route for finding lower bounds for property testing.

\paragraph*{Structure of the paper.}
We begin with the preliminaries in Section~\ref{sec:preliminaries}. In 
Section~\ref{sec:distanceTOHamiltonicity} we introduce \emph{local} Hamiltonicity, discuss \emph{distance} to Hamiltonicity, and we provide our construction. The construction takes a $d$-regular graph and turns it into a graph of degree at most $d+3$ with additional properties. In Section~\ref{sec:farness} we prove that there is a small $\epsilon$ such that any family of graphs obtained via the construction is $\epsilon$-far from being Hamiltonian.
Section~\ref{sec:localHAM} then shows that if we start our construction with $d$-regular expander graphs,
we obtain a family that is locally Hamiltonian. In Section~\ref{sec:applicationsToPT} we prove a known lower-bound for property testing from our construction.
\section{Preliminaries}\label{sec:preliminaries}
Let $\mathbb{N}$ denote the set of natural numbers including $0$. We denote $\mathbb{N}_{\geq n}:=\{m\in \mathbb{N}: m\geq n \}$ and $[n]:=\{1,\dots,n\}$ for any $n\in \mathbb{N}$ (where we let $[0]:=\emptyset$). For two sets $A,B$ we use $A\triangle B$ to denote the symmetric difference of $A$ and $B$.

This paper concerns simple undirected graphs, however, we will use directed graphs in our construction. Unless otherwise specified graphs are undirected. 

An \emph{undirected graph} $G$ is a tuple $(V(G),E(G))$ where $V(G)$ is a finite set of vertices and $E(G)\subseteq \{e\subseteq V(G): |e|=2\}$ is the set of edges. A \emph{directed graph} $G$ is a tuple $(V(G),E(G))$ where $V(G)$ is a finite  set of vertices and $E(G)\subseteq V\times V$ is the set of edges. For a directed graph $G$ and a vertex $v\in V(G)$ we denote the set of all incoming edges of $v$ by  $E_G^{-}(v)$ and the set of all outgoing edges of $v$ by $E_G^{+}(v)$. The \emph{order} of a graph $G$ is the size of $V(G)$. 

An \emph{isomorphism} from a graph $G$ to a graph $H$ is a bijective map $f:V(G)\rightarrow V(H)$ which preserves the edge relation, \ie\hspace{-4pt},  $\{v,w\}\in E(G)$ iff $\{f(v),f(w)\}\in E(H)$. Equivalently an \emph{isomorphism} from a directed graph $G$ to a directed graph $H$ is a bijective map $f:V(G)\rightarrow V(H)$ such that $(v,w)\in E(G)$ iff $(f(v),f(w))\in E(H)$. Two graphs $G,H$  are called \emph{isomorphic}, denoted by $G\cong H$, if there is an isomorphism between them. A graph $H$ is a \emph{subgraph} of a graph $G$ if $V(H)\subseteq V(G)$ and $E(H)\subseteq E(G)$. For any graph $G$ and  $S\subseteq V(G)$ we let $G[S]:=(S, \{\{v,w\}\in E(G): v,w\in S \})$ be the subgraph of $G$ induced by $S$. We call a subgraph $H$ of $G$ an \emph{induced subgraph} of $G$ if $H$ is the subgraph of $G$ induced by some set $S\subseteq V(G)$. For a graph $G$ and vertices $v,w\in V(G)$ we say that $v$ is  a neighbour of $w$ or that $v$ is adjacent to $w$ if $\{v,w\}\in E(G)$.  For $S\subseteq V(G)$ we define the  \emph{neighbourhood} of $S$ in $G$, denoted $N_G(S)$ to be the set of vertices $S\cup \{v\in V(G): v \text{ is a neighbour of some }w\in S\}$. This notion of neighbourhood is often referred to as the \emph{closed neighbourhood}. 

For a graph $G$ (directed or undirected) the \emph{degree of a vertex} $v\in V(G)$, denoted  $\deg_G(v)$, is the number of edges that contain vertex $v$. The \emph{degree of a graph} $G$, denoted $\deg(G)$, is the maximum degree over all vertices.  A graph is called \emph{$d$-regular} if every vertex $v\in V(G)$ has degree $d$, where $d\in \mathbb{N}$. A graph has bounded degree $d$ if $\deg(G)\leq d$, where $d\in \mathbb{N}$. For $d\in \mathbb{N}$ we denote the class of all bounded degree $d$ graphs by $\mathcal{C}_d$.

A \emph{path}  of length $\ell$ in a graph $G$ (undirected or directed) is a sequence $(p_0,p_1,\dots,p_\ell)$ of vertices of $G$ such that $\{p_{i-1},p_{i}\}\in E(G)$/$(p_{i-1},p_{i})\in E(G)$ for $i\in [\ell]$.
A \emph{simple path} in $G$ is a path in which no vertex appears twice.
A \emph{cycle} is a path $C=(c_0,\dots,c_\ell)$ such that $c_0=c_\ell$ and $(c_1,\dots,c_\ell)$ is a simple path.
A \emph{Hamiltonian cycle} is a cycle which contains every vertex of $G$.
We call $G$ Hamiltonian if $G$ contains a Hamiltonian cycle.
A path $P'=(p_1',\dots,p_{\ell'}')$ is a \emph{subpath} of a cycle $C=(c_0,\dots,c_\ell)$ if there is an index $i\in [\ell]$ such that either $p_j'=p_{(i+j\mod \ell)}$ for every $j\in [\ell']$ or $p_j'=p_{(i+\ell'-j\mod \ell)}$ for every $j\in [\ell']$. Note that this means that subpaths appear either in the path or in the reversed path.
We choose this definition of subpath for convenient notation below.

For a graph  $G$ we define the \emph{expansion ratio} to be 
\begin{displaymath}
h(G):=\min_{\{S\subset V(G): |S|\leq |V(G)|/2\}}\frac{|\{e\in E(G): |e\cap S|=1\}|}{|S|}.
\end{displaymath}
For $d\in \mathbb{N}$ and any constant $\epsilon>0$ we call a sequence $(G_m)_{m\in\mathbb{N}}$ of $d$-regular graphs of increasing number of vertices a \emph{family of $\epsilon$-expanders} if  $h(G_m)\geq \epsilon$ for all $m\in \mathbb{N}$.

\section{Local Hamiltonicity and distance to Hamiltonicity}\label{sec:distanceTOHamiltonicity}
In this section we introduce the central concepts in this paper and explain our construction. The proofs of the central properties of the construction are given in the next sections.
\begin{definition}[$\epsilon$-farness from being Hamiltonian]\label{def:farFromHam}
	Let $d\in \mathbb{N}$ and $\epsilon\in [0,1]$.
	A graph   $G\in \mathcal{C}_d$  is $\epsilon$-far from being Hamiltonian if for every set $E\subseteq \{e\subseteq V(G): |e|=2\}$ of size less than or equal to $\epsilon d\cdot |V(G)|$ the graph $(V(G),E(G)\triangle E)$ is not Hamiltonian. 
\end{definition}
\begin{definition}[Locally Hamiltonian]\label{def:locallyHam}
	Let $\mathcal{C}$ be a class of graphs and let $\delta\in (0,1]$.
	A graph $G\in \mathcal{C}$ is called \emph{$\delta$-locally Hamiltonian on $\mathcal{C}$} if for every set $S\subseteq V(G)$ of at most $\delta\cdot|V(G)|$ vertices there is a Hamiltonian graph $H:=H_S\in \mathcal{C}$ with $|V(H)|=|V(G)|$, a subset $T:=T_S\subseteq V(H)$ and an isomorphism from $G[N_G(S)]$ to $H[N_H(T)]$ which maps $S$ onto $T$.  
\end{definition}
Note that by relaxing $|V(G)|=|V(H)|$  to $|V(H)|>|N_G(S)|$ we get an equivalent definition. As long as the Hamiltonian cycle in $H$ contains at least one edge which is not in $H[N_H(T)]$ we can contract or subdivide edges on the Hamiltonian cycle of $H$ that are not in $H[N_H(T)]$ to make $|V(H)|=|V(G)|$.

\begin{remark}
	Let $\mathcal C$ be a graph class. Every Hamiltonian graph in $\mathcal C$ is $\delta$-locally Hamiltonian for every $\delta\in (0,1]$. And every graph $G\in \mathcal C$ is $1$-locally Hamiltonian iff $G$ is Hamiltonian.
	
	Let $d\geq 2$. A graph $G$ on $n\geq 2d$
	vertices is $1/n$-locally Hamiltonian on $\mathcal C_d$ iff the minimum degree of $G$ is greater than $1$. The `only if' part is easy to see. The `if' part is true because the closed neighbourhood of any vertex $v$ contains at most $d+1$ vertices and contains at least one path of length two. This path of length two can be completed into a Hamiltonian cycle without adding edges within the neighbourhood of $v$, because $G$ contains at least $d-1$ further vertices, via which we connect the vertices and the path of length two in the neighbourhood of $v$. 
\end{remark}
The next lemma states that if  $G\in \mathcal{C}_d$ has many subsets of vertices whose neighbourhoods witness non-Hamiltonicity, then it is far from being Hamiltonian. Here we say that the neighbourhood of $S\subseteq V(G)$ \emph{witnesses} non-Hamiltonicity if for every  Hamiltonian $H\in \mathcal{C}_d$ of the same order as $G$ and every $T\subseteq V(H)$ there is no isomorphism from $G[N_G(S)]$ to $H[N_H(T)]$ that maps $S$ onto $T$. For example the neighbourhood of every vertex of degree $1$ in a graph $G$ witnesses non-Hamiltonicity of $G$ and the neighbourhood of any connected component of $G$ (unless $G$ is connected) witnesses non-Hamiltonicity of $G$.
\begin{lemma}\label{lem:manyCounterExamplesImplyFarness}
	Let $d\in \mathbb{N}_{\geq 2}$ and $G\in \mathcal{C}_d$. For all $0\leq \epsilon < 1/2d$, there is a number
	$\lambda=\lambda(\epsilon,d)\in (0,1)$ such that if there are $\lambda n$ subsets of $V(G)$ whose (closed) neighbourhoods are pairwise disjoint and each witnesses
	non-Hamiltonicity, then $G$ is $\epsilon$-far from being Hamiltonian.
\end{lemma}
\begin{proof}
	Let $\lambda(\epsilon,d):=2d\epsilon$.
	First note that if a set $S\subseteq V(G)$ witnesses non-Hamiltonicity then every set $E\subseteq \{e\subseteq V(G): |e|=2\}$, for which $(V(G),E(G)\triangle E)$ is Hamiltonian, must contain $e$ such that $S\cap e \not= \emptyset$.
	Since the $\lambda n$ neighbourhoods of sets witnessing non-Hamiltonicity are pairwise disjoint we get that  the size of every set $E\subseteq \{e\subseteq V(G): |e|=2\}$, for which $(V(G),E(G)\triangle E)$ is Hamiltonian, is at least $\lambda n/2=\epsilon d n$ and hence $G$ is $\epsilon$-far from being Hamiltonian.
\end{proof}
Note that for the caterpillar (see Figure~\ref{fig:examplesIntroduction}) such pairwise disjoint sets of vertices, whose neighbourhoods witness non-Hamiltonicity, are the singleton sets consisting of the vertices of degree one. Similarly, for the cycle of $C_4$'s (see Figure~\ref{fig:examplesIntroduction}), 
we can choose the vertex set of every other $C_4$ on the cycle.

One might wonder if the converse of Lemma~\ref{lem:manyCounterExamplesImplyFarness} is true, \ie\hspace{-4pt},
if $G$ is $\epsilon$-far from being Hamiltonian, then $G$ contains a linear fraction of pairwise disjoint sets of vertices whose neighbourhoods
witness non-Hamiltonicity. (This would actually imply the existence of a one-sided error property tester with constant query complexity.)
Our examples below show that this is not the case. In fact it even shows  that for some constant $c\in \mathbb{N}$ there is a class of graphs which are $\epsilon$-far from being Hamiltonian but we cannot even find $c$ sets of vertices with pairwise disjoint neighbourhoods  witnessing non-Hamiltonicity. In other words there is a class of graphs which are $\epsilon$-far from being Hamiltonian but $\delta$-locally Hamiltonian for some $\delta,\epsilon\in (0,1)$.

\subsection{Construction}\label{sec:construction} 
In this section we introduce the main step of our construction of graphs which are locally Hamiltonian and far from being Hamiltonian. At a high level, we construct a graph $G_\mathcal{E}$ by choosing a $d$-regular base graph $\mathcal{E}$ and building $G_\mathcal{E}$ by introducing a path-gadget for every edge of $\mathcal{E}$, connecting these path-gadgets into a large cycle and linking path gadgets together if the  edges of $\mathcal{E}$ corresponding to the path gadgets are incident to the same vertex. We give the precise construction in the following. 

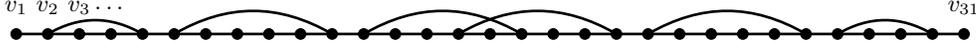
\begin{figure}
	\centering
	\begin{tikzpicture}[>=stealth']
	\definecolor{C1}{RGB}{1,1,1}
	\definecolor{C2}{RGB}{0,204,204}
	\definecolor{C3}{RGB}{204,0,0}
	\definecolor{C4}{RGB}{0,204,0}
	
	\tikzstyle{ns1}=[line width=1]

	\def \breadthDist {0.42}
	\foreach \x in {1,...,31} {
		\node[draw,circle,fill=black,inner sep=0pt, minimum width=4pt] (thisNode) at (\x*\breadthDist,0) {};
	}
	\foreach \x in {1,...,30} {
		\pgfmathtruncatemacro{\xx}{\x+1};
		\path[C1,ns1]          (\x*\breadthDist,0)  edge  (\xx*\breadthDist,0);
	}
	\foreach \x in {2,27} {
		\pgfmathtruncatemacro{\xx}{\x+3};
		\path[C1,ns1]          (\x*\breadthDist,0)  edge   [bend left] (\xx*\breadthDist,0);
	}
	\foreach \x in {6,12,15,21} {
		\pgfmathtruncatemacro{\xx}{\x+5};
		\path[C1,ns1]          (\x*\breadthDist,0)  edge   [bend left] (\xx*\breadthDist,0);
	}

	\node[minimum height=10pt,inner sep=0,font=\small,C1] at (\breadthDist,0.35) {$v_1$};
	\node[minimum height=10pt,inner sep=0,font=\small,C1] at (2*\breadthDist,0.35) {$v_2$};
	\node[minimum height=10pt,inner sep=0,font=\small,C1] at (3*\breadthDist,0.35) {$v_3$};
	\node[minimum height=10pt,inner sep=0,font=\small,C1] at (4*\breadthDist,0.35) {$\dots$};
	\node[minimum height=10pt,inner sep=0,font=\small,C1] at (31*\breadthDist,0.35) {$v_{31}$};

	\end{tikzpicture}
	\caption{$P(v_1,\dots,v_{31})$.} \label{fig:P}
\end{figure}
First we create a gadget (see Figure~\ref{fig:P} for illustration). Let $v_1,\dots,v_{31}$ be a set of vertices. Then we let $P(v_1,\dots,v_{31})$ be the graph with vertex set $\{v_1,\dots,v_{31}\}$ and edge set
\begin{align*}
\big\{\{v_i,v_{i+1}\},\{v_j,v_{j+3}\},\{v_k,v_{k+5}\}: i\in \{1,\dots,30\},j\in \{2,27\},k\in \{6,12,15,21\} \big\}.
\end{align*}

For a graph $G$ with $\{u_1,\dots,u_{31},v_1,\dots,v_{31},w_1,\dots,w_6\}\subseteq V(G)$, $G[u_1,\dots,u_{31}]=P(u_1,\dots,u_{31})$ and $G[v_1,\dots,v_{31}]=P(v_1,\dots,v_{31})$ we say that $G$ contains a \emph{link from $P(u_1,\dots,u_{31})$ to $P(v_1,\dots,v_{31})$ via $w_1,\dots,w_6$} (see Figure~\ref{fig:link} for illustration), if $E(G)$ contains 
\begin{align*}
\Big\{\{u_{23},v_3\}, \{u_{18},v_{8}\}, \{u_{29},v_{9}\}, \{u_{24},v_{14}\}, &\{v_5,w_1\},\{w_1,w_2\}, \{w_2,w_3\},\{w_3,u_{23}\},\\ &\{u_{24},w_4\},\{w_4,w_5\},\{w_5,w_6\}, \{w_6,v_{12}\}\Big\}.
\end{align*}

\begin{figure}
	\centering
	\begin{tikzpicture}[>=stealth']
	\definecolor{C1}{RGB}{1,1,1}
	\definecolor{C2}{RGB}{0,204,204}
	\definecolor{C3}{RGB}{204,0,0}
	\definecolor{C4}{RGB}{0,204,0}
	
	\tikzstyle{ns1}=[line width=1]
	
	\def \heightDist {1.6}
	\def \breadthDist {0.7}
	\def \heightShift {1}
	\node[draw,circle,fill=black,inner sep=0pt, minimum width=4pt] (1) at (-9.5*\breadthDist,\heightDist) {};
	\node[draw,circle,fill=black,inner sep=0pt, minimum width=4pt] (2) at (-8.5*\breadthDist,\heightDist) {};
	\node[draw,circle,fill=black,inner sep=0pt, minimum width=4pt] (3) at (-7.5*\breadthDist,\heightDist) {};
	\node[draw,circle,fill=black,inner sep=0pt, minimum width=4pt] (4) at (-6.5*\breadthDist,\heightDist) {};
	\node[draw,circle,fill=black,inner sep=0pt, minimum width=4pt] (5) at (-5.5*\breadthDist,\heightDist) {};
	\node[draw,circle,fill=black,inner sep=0pt, minimum width=4pt] (6) at (-4.5*\breadthDist,\heightDist) {};
	\node[draw,circle,fill=black,inner sep=0pt, minimum width=4pt] (7) at (-3.5*\breadthDist,\heightDist) {};
	\node[draw,circle,fill=black,inner sep=0pt, minimum width=4pt] (8) at (-2.5*\breadthDist,\heightDist) {};
	\node[draw,circle,fill=black,inner sep=0pt, minimum width=4pt] (9) at (-1.5*\breadthDist,\heightDist) {};
	\node[draw,circle,fill=black,inner sep=0pt, minimum width=4pt] (10) at (-0.5*\breadthDist,\heightDist) {};
	\node[draw,circle,fill=black,inner sep=0pt, minimum width=4pt] (11) at (0.5*\breadthDist,\heightDist) {};
	\node[draw,circle,fill=black,inner sep=0pt, minimum width=4pt] (12) at (1.5*\breadthDist,\heightDist) {};
	\node[draw,circle,fill=black,inner sep=0pt, minimum width=4pt] (13) at (2.5*\breadthDist,\heightDist) {};
	\node[draw,circle,fill=black,inner sep=0pt, minimum width=4pt] (14) at (3.5*\breadthDist,\heightDist) {};
	\node[draw,circle,fill=black,inner sep=0pt, minimum width=4pt] (15) at (4.5*\breadthDist,\heightDist) {};
	\node[draw,circle,fill=black,inner sep=0pt, minimum width=4pt] (16) at (5.5*\breadthDist,\heightDist) {};
	\node[draw,circle,fill=black,inner sep=0pt, minimum width=4pt] (17) at (6.5*\breadthDist,\heightDist) {};		
	
	\node[draw,circle,fill=black,inner sep=0pt, minimum width=4pt] (15') at (-10.5*\breadthDist,0) {};
	\node[draw,circle,fill=black,inner sep=0pt, minimum width=4pt] (16') at (-9.5*\breadthDist,0) {};
	\node[draw,circle,fill=black,inner sep=0pt, minimum width=4pt] (17') at (-8.5*\breadthDist,0) {};
	\node[draw,circle,fill=black,inner sep=0pt, minimum width=4pt] (18') at (-7.5*\breadthDist,0) {};
	\node[draw,circle,fill=black,inner sep=0pt, minimum width=4pt] (19') at (-6.5*\breadthDist,0) {};
	\node[draw,circle,fill=black,inner sep=0pt, minimum width=4pt] (20') at (-5.5*\breadthDist,0) {};
	\node[draw,circle,fill=black,inner sep=0pt, minimum width=4pt] (21') at (-4.5*\breadthDist,0) {};
	\node[draw,circle,fill=black,inner sep=0pt, minimum width=4pt] (22') at (-3.5*\breadthDist,0) {};
	\node[draw,circle,fill=black,inner sep=0pt, minimum width=4pt] (23') at (-2.5*\breadthDist,0) {};
	\node[draw,circle,fill=black,inner sep=0pt, minimum width=4pt] (24') at (-1.5*\breadthDist,0) {};
	\node[draw,circle,fill=black,inner sep=0pt, minimum width=4pt] (25') at (-0.5*\breadthDist,0) {};
	\node[draw,circle,fill=black,inner sep=0pt, minimum width=4pt] (26') at (0.5*\breadthDist,0) {};
	\node[draw,circle,fill=black,inner sep=0pt, minimum width=4pt] (27') at (1.5*\breadthDist,0) {};
	\node[draw,circle,fill=black,inner sep=0pt, minimum width=4pt] (28') at (2.5*\breadthDist,0) {};
	\node[draw,circle,fill=black,inner sep=0pt, minimum width=4pt] (29') at (3.5*\breadthDist,0) {};
	\node[draw,circle,fill=black,inner sep=0pt, minimum width=4pt] (30') at (4.5*\breadthDist,0) {};
	\node[draw,circle,fill=black,inner sep=0pt, minimum width=4pt] (31') at (5.5*\breadthDist,0) {};
	\node[draw,circle,fill=black,inner sep=0pt, minimum width=4pt] (1'') at (-5.5*\breadthDist+0.25*3*\breadthDist,0.75*\heightDist) {};
	\node[draw,circle,fill=black,inner sep=0pt, minimum width=4pt] (2'') at (-5.5*\breadthDist+0.5*3*\breadthDist,0.5*\heightDist) {};
	\node[draw,circle,fill=black,inner sep=0pt, minimum width=4pt] (3'') at (-5.5*\breadthDist+0.75*3*\breadthDist,0.25*\heightDist) {};
	\node[draw,circle,fill=black,inner sep=0pt, minimum width=4pt] (4'') at (-1.5*\breadthDist+0.25*3*\breadthDist,0.25*\heightDist) {};
	\node[draw,circle,fill=black,inner sep=0pt, minimum width=4pt] (5'') at (-1.5*\breadthDist+0.5*3*\breadthDist,0.5*\heightDist) {};
	\node[draw,circle,fill=black,inner sep=0pt, minimum width=4pt] (6'') at (-1.5*\breadthDist+0.75*3*\breadthDist,0.75*\heightDist) {};
	\path[C1,ns1]          (1)  edge  (2);
	\path[C1,ns1]          (2)  edge  (3);
	\path[C1,ns1]          (3)  edge  (4);
	\path[C1,ns1]          (4)  edge  (5);
	\path[C1,ns1]          (5)  edge  (6);
	\path[C1,ns1]          (6)  edge  (7);
	\path[C1,ns1]          (7)  edge  (8);
	\path[C1,ns1]          (8)  edge  (9);
	\path[C1,ns1]          (9)  edge  (10);
	\path[C1,ns1]          (10)  edge  (11);
	\path[C1,ns1]          (11)  edge  (12);
	\path[C1,ns1]          (12)  edge  (13);
	\path[C1,ns1]          (13)  edge  (14);
	\path[C1,ns1]          (14)  edge  (15);
	\path[C1,ns1]          (15)  edge  (16);
	\path[C1,ns1]          (16)  edge  (17);
	\path[C1,ns1]          (15')  edge  (16');
	\path[C1,ns1]          (16')  edge  (17');
	\path[C1,ns1]          (17')  edge  (18');
	\path[C1,ns1]          (18')  edge  (19');
	\path[C1,ns1]          (19')  edge  (20');
	\path[C1,ns1]          (20')  edge  (21');
	\path[C1,ns1]          (21')  edge  (22');
	\path[C1,ns1]          (22')  edge  (23');
	\path[C1,ns1]          (23')  edge  (24');
	\path[C1,ns1]          (24')  edge  (25');
	\path[C1,ns1]          (25')  edge  (26');
	\path[C1,ns1]          (26')  edge  (27');
	\path[C1,ns1]          (27')  edge  (28');
	\path[C1,ns1]          (28')  edge  (29');
	\path[C1,ns1]          (29')  edge  (30');
	\path[C1,ns1]          (30')  edge  (31');
	\path[C1,ns1]          (3)  edge  (23');
	\path[C1,ns1]          (18')  edge  (8);
	\path[C1,ns1]          (9)  edge  (29');
	\path[C1,ns1]          (24')  edge  (14);
	\path[C1,ns1]          (5)  edge  (1'');
	\path[C1,ns1]          (1'')  edge  (2'');
	\path[C1,ns1]          (2'')  edge  (3'');
	\path[C1,ns1]          (3'')  edge  (23');
	\path[C1,ns1]          (24')  edge  (4'');
	\path[C1,ns1]          (4'')  edge  (5'');
	\path[C1,ns1]          (5'')  edge  (6'');
	\path[C1,ns1]          (6'')  edge  (12);
	\path[C1,ns1]          (2)  edge   [bend left] (5);
	\path[C1,ns1]          (6)  edge   [bend left] (11);
	\path[C1,ns1]          (12)  edge   [bend left] (17);
	\path[C1,ns1]          (15)  edge   [bend left] (7*\breadthDist,0.3*\heightShift+\heightDist);
	\path[C1,ns1]          (-11*\breadthDist,-0.3*\heightShift)  edge   [bend right] (17');
	\path[C1,ns1]          (15')  edge   [bend right] (20');
	\path[C1,ns1]          (21')  edge   [bend right] (26');
	\path[C1,ns1]          (27')  edge   [bend right] (30');
	
	\node[minimum height=10pt,inner sep=0,font=\small,C1] at (-9.5*\breadthDist,\heightDist+0.3) {$v_1$};
	\node[minimum height=10pt,inner sep=0,font=\small,C1] at (-7.5*\breadthDist,\heightDist-0.3) {$v_{3}$};
	\node[minimum height=10pt,inner sep=0,font=\small,C1] at (-5.5*\breadthDist,\heightDist+0.3) {$v_{5}$};
	\node[minimum height=10pt,inner sep=0,font=\small,C1] at (-2.5*\breadthDist,\heightDist-0.3) {$v_{8}$};
	\node[minimum height=10pt,inner sep=0,font=\small,C1] at (-1.5*\breadthDist,\heightDist-0.3) {$v_{9}$};
	\node[minimum height=10pt,inner sep=0,font=\small,C1] at (1.5*\breadthDist,\heightDist+0.3) {$v_{12}$};
	\node[minimum height=10pt,inner sep=0,font=\small,C1] at (3.5*\breadthDist,\heightDist-0.3) {$v_{14}$};
	\node[minimum height=10pt,inner sep=0,font=\small,C1] at (-5.5*\breadthDist+0.25*3*\breadthDist+0.3,0.75*\heightDist) {$w_{1}$};
	\node[minimum height=10pt,inner sep=0,font=\small,C1] at (-5.5*\breadthDist+0.5*3*\breadthDist+0.3,0.5*\heightDist) {$w_{2}$};
	\node[minimum height=10pt,inner sep=0,font=\small,C1] at (-5.5*\breadthDist+0.75*3*\breadthDist+0.3,0.25*\heightDist) {$w_{3}$};
	\node[minimum height=10pt,inner sep=0,font=\small,C1] at (-1.5*\breadthDist+0.25*3*\breadthDist-0.3,0.25*\heightDist) {$w_{4}$};
	\node[minimum height=10pt,inner sep=0,font=\small,C1] at (-1.5*\breadthDist+0.5*3*\breadthDist-0.3,0.5*\heightDist) {$w_{5}$};
	\node[minimum height=10pt,inner sep=0,font=\small,C1] at (-1.5*\breadthDist+0.75*3*\breadthDist-0.3,0.75*\heightDist) {$w_{6}$};
	\node[minimum height=10pt,inner sep=0,font=\small,C1] at (5.5*\breadthDist,-0.3) {$u_{31}$};
	\node[minimum height=10pt,inner sep=0,font=\small,C1] at (3.5*\breadthDist,0.3) {$u_{29}$};
	\node[minimum height=10pt,inner sep=0,font=\small,C1] at (-1.5*\breadthDist,-0.3) {$u_{24}$};
	\node[minimum height=10pt,inner sep=0,font=\small,C1] at (-2.5*\breadthDist,-0.3) {$u_{23}$};
	\node[minimum height=10pt,inner sep=0,font=\small,C1] at (-7.5*\breadthDist,0.3) {$u_{18}$};

	\end{tikzpicture}
	\caption{A link from $P(u_1,\dots,u_{31})$ to $P(v_1,\dots,v_{31})$ via $w_1,\dots,w_6$. }\label{fig:link}
\end{figure}
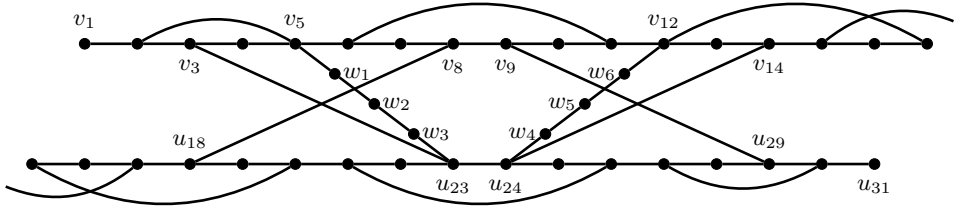
Finally to any graph $G$ we associate a directed graph $\vec{G}$ which is the graph that is obtained from $G$ by replacing every edge $\{u,v\}\in E(G)$ by the two directed edges $(u,v)$ and $(v,u)$. We can now define the graph construction.
\begin{definition}\label{def:G_E}
	Let $\mathcal{E}$ be a $d$-regular graph (the \emph{base graph}) and $f:E(\vec{\mathcal{E}})\rightarrow \{1,\dots,|E(\vec{\mathcal{E}})|\}$ be any linear order on  $E(\vec{\mathcal{E}})$. We define the graph $G_\mathcal{E}$ as follows.
	\begin{displaymath}
	V(G_\mathcal{E}):=\{a_1^e,\dots,a_{31}^e: e\in E(\vec{\mathcal{E}})\}\cup \{b_1^v,\dots,b_6^v: v\in V(\vec{\mathcal{E}})\}.
	\end{displaymath}
	$E(G_\mathcal{E})$ consists of the minimum set of edges such that 
	\begin{itemize}
		\item $G_\mathcal{E}[a_1^e,\dots,a_{31}^e]=P(a_1^e,\dots,a_{31}^e)$ for every $e\in E(\vec{\mathcal{E}})$,
		\item $a_{31}^{f^{-1}(i)}$ is adjacent to $a_1^{f^{-1}(j)}$ for every $i\in [|E(\vec{\mathcal{E}})|]$, $j:=1$ if $i=|E(\vec{\mathcal{E}})|$ and $j:=i+1$ otherwise and 
		\item $G_\mathcal{E}$ contains a link from 
		$P(a_1^{(u,v)},\dots,a_{31}^{(u,v)})$ to $P(a_1^{(v,w)},\dots,a_{31}^{(v,w)})$ via $b_1^v,\dots,b_6^v$ for every triple of vertices $u,v,w\in V(\vec{\mathcal{E}})$ with $(u,v),(v,w)\in E(\vec{\mathcal{E}})$.
	\end{itemize}
\end{definition}
See Figure~\ref{fig:close-up} for an illustration. We would like to point out that the minimality condition on the set of edges of $G_\mathcal{E}$ in the above definition is necessary, as omitting it would allow adding of any number of edges between different path gadgets. 

Note that the construction of $G_\mathcal{E}$ depends on  $f$ as well as $\mathcal{E}$, but since being locally Hamiltonian and being far from Hamiltonian  are independent of which linear order $f$ we use, we omit the dependency on~$f$.

\begin{remark}\label{rem:degreeAndNumberOfVerticesOFG_E}
	If $\mathcal{E}$ is $d$-regular, for $d\geq 1$, and $|V(\mathcal{E})|=n$, then the degree of $G_\mathcal{E}$ is at most $d+3$ and $|V(G_\mathcal{E})|=(6+31d)n$.
\end{remark}
\begin{note}
	$G_\mathcal{E}$ contains a large cycle of length $31dn$, \ie\hspace{-4pt}, the cycle
	\begin{displaymath}
	(\dots\dots, a_{31}^{f^{-1}(i-1)},a_1^{f^{-1}(i)},a_2^{f^{-1}(i)},\dots,a_{31}^{f^{-1}(i)},a_1^{f^{-1}(i+1)},\dots\dots).
	\end{displaymath}
	However $G_\mathcal{E}$ also contains $6n$ vertices which are not part of this cycle.
\end{note}

\section{The construction is far from being Hamiltonian}\label{sec:farness}
In this section we prove the following.
\begin{theorem}\label{thm:farFromHam}
	For every $d\in \mathbb{N}_{\geq 2}$ there is $\epsilon=\epsilon(d)\in (0,1)$ such that for any $d$-regular graph  $\mathcal{E}$  the graph $G_\mathcal{E}$ constructed in \ref{def:G_E} is $\epsilon$-far from being Hamiltonian.
\end{theorem}
To prove Theorem~\ref{thm:farFromHam} we use the two technical lemmas below (Lemma~\ref{lem:orderOfBeginningMiddleAndEndPieceOfP(...)} and Lemma~\ref{lem:inAndOutSetEqualCardinality}). 
They will be applied to graphs $G$ obtained from $G_\mathcal{E}$ by modifying a small fraction of the edges of $G_\mathcal{E}$. 
Therefore they are stated for graphs $G$ which share certain induced subgraphs with $G_\mathcal{E}$. 
The first of the two lemmas (Lemma~\ref{lem:orderOfBeginningMiddleAndEndPieceOfP(...)}) states that if $G$ has a Hamiltonian cycle and a certain induced subgraph, which also appears in $G_\mathcal{E}$, then the Hamiltonian cycle has certain subpaths. 
The proof of Lemma~\ref{lem:orderOfBeginningMiddleAndEndPieceOfP(...)} is illustrated in Figure~\ref{fig:close-up}. We will use the following easy observation in the proof of Lemma~\ref{lem:orderOfBeginningMiddleAndEndPieceOfP(...)}.
\begin{remark}\label{rem:vertexOfDegreeTwo}
	Let $G$ be a graph, $u\in V(G)$ a vertex of degree $2$ and $v,w$ the two neighbours of $u$. Then any cycle $C$ containing the vertex $u$ must contain $(v,u,w)$ as a subpath.
\end{remark}
Recall that subpaths appear either in the path or in the reversed path.

\begin{lemma}\label{lem:orderOfBeginningMiddleAndEndPieceOfP(...)}
	Let $\mathcal{E}$ be any $d$-regular graph and  $G_\mathcal{E}$ as defined in Definition~\ref{def:G_E}. Pick $v\in V(\vec{\mathcal{E}})$ and let $S_v:=\{a_1^e,\dots,a_{31}^e: e\in E(\vec{\mathcal{E}}),e\text{ contains }v\}\cup \{b_1^v,\dots,b_6^v\}$. 
	Let $G$ be a graph with $S_v\subseteq V(G)$. Assume  $G_\mathcal{E}[N_{G_\mathcal{E}}(S_v)]\cong G[N_G(S_v)]$ and $f: S_v\rightarrow S_v$ defined by $f(v)=v$ for $v\in S_v$ is an isomorphism from $G_\mathcal{E}[S_v]$ to $G[S_v]$.
	Then for every Hamiltonian cycle $C$ in $G$ and every edge $e\in V(\vec{\mathcal{E}})$ incident to $v$ the following properties hold.
	\begin{enumerate}
		\item\label{order1} Either $(a_1^e,\dots,a_5^e)$   or $(a_1^e,a_2^e,a_5^e,a_4^e,a_3^e)$ is a   subpath of $C$.
		\item\label{order2} Either $(a_{27}^e,\dots,a_{31}^e)$  or $(a_{29}^e,a_{28}^e,a_{27}^e,a_{30}^e,a_{31}^e)$ is a   subpath of $C$. 
		\item\label{order3} Either $(a_{12}^e,\dots,a_{20}^e)$   or  $(a_{14}^e,a_{13}^e,a_{12}^e,a_{17}^e,a_{16}^e,a_{15}^e,a_{20}^e,a_{19}^e,a_{18}^e)$ is a   subpath of $C$.
		\item\label{order4} If $e\in E_G^{+}(v)$  then either $(a_6^e,\dots,a_{11}^e)$  or $(a_8^e,a_7^e,a_6^e,a_{11}^e,a_{10}^e,a_{9}^e)$ is a   subpath of $C$.
		\item\label{order5} If  $e\in E_G^{-}(v)$ then either $(a_{21}^e,\dots,a_{26}^e)$  or $(a_{23}^e,a_{22}^e,a_{21}^e,a_{26}^e,a_{25}^e,a_{24}^e)$ is a   subpath of $C$. 
	\end{enumerate}
\end{lemma}
\begin{figure}
	\centering
	\begin{tikzpicture}[>=stealth']
	\definecolor{C1}{RGB}{1,1,1}
	\definecolor{C2}{RGB}{0,204,204}
	\definecolor{C3}{RGB}{204,0,0}
	\definecolor{C4}{RGB}{0,204,0}
	
	\tikzstyle{ns1}=[line width=1]
	
	\def \heightDist {1.8}
	\def \breadthDist {0.42}
	\def \heightShift {1}
	\def \intensity {50}
	\def \rad {0.35}
	\def \radius {0.5}
	\def \Rad {0.65}
	\def \Radius {1}
	\node[draw,circle,fill=black,inner sep=0pt, minimum width=4pt] (1) at (-9.5*\breadthDist,\heightDist) {};
	\node[draw,circle,fill=black,inner sep=0pt, minimum width=4pt] (2) at (-8.5*\breadthDist,\heightDist) {};
	\node[draw,circle,fill=black,inner sep=0pt, minimum width=4pt] (3) at (-7.5*\breadthDist,\heightDist) {};
	\node[draw,circle,fill=black,inner sep=0pt, minimum width=4pt] (4) at (-6.5*\breadthDist,\heightDist) {};
	\node[draw,circle,fill=black,inner sep=0pt, minimum width=4pt] (5) at (-5.5*\breadthDist,\heightDist) {};
	\node[draw,circle,fill=black,inner sep=0pt, minimum width=4pt] (6) at (-4.5*\breadthDist,\heightDist) {};
	\node[draw,circle,fill=black,inner sep=0pt, minimum width=4pt] (7) at (-3.5*\breadthDist,\heightDist) {};
	\node[draw,circle,fill=black,inner sep=0pt, minimum width=4pt] (8) at (-2.5*\breadthDist,\heightDist) {};
	\node[draw,circle,fill=black,inner sep=0pt, minimum width=4pt] (9) at (-1.5*\breadthDist,\heightDist) {};
	\node[draw,circle,fill=black,inner sep=0pt, minimum width=4pt] (10) at (-0.5*\breadthDist,\heightDist) {};
	\node[draw,circle,fill=black,inner sep=0pt, minimum width=4pt] (11) at (0.5*\breadthDist,\heightDist) {};
	\node[draw,circle,fill=black,inner sep=0pt, minimum width=4pt] (12) at (1.5*\breadthDist,\heightDist) {};
	\node[draw,circle,fill=black,inner sep=0pt, minimum width=4pt] (13) at (2.5*\breadthDist,\heightDist) {};
	\node[draw,circle,fill=black,inner sep=0pt, minimum width=4pt] (14) at (3.5*\breadthDist,\heightDist) {};
	\node[draw,circle,fill=black,inner sep=0pt, minimum width=4pt] (15) at (4.5*\breadthDist,\heightDist) {};
	\node[draw,circle,fill=black,inner sep=0pt, minimum width=4pt] (16) at (5.5*\breadthDist,\heightDist) {};
	\node[draw,circle,fill=black,inner sep=0pt, minimum width=4pt] (17) at (6.5*\breadthDist,\heightDist) {};
	\node[draw,circle,fill=black,inner sep=0pt, minimum width=4pt] (18) at (7.5*\breadthDist,\heightDist) {};
	\node[draw,circle,fill=black,inner sep=0pt, minimum width=4pt] (19) at (8.5*\breadthDist,\heightDist) {};
	\node[draw,circle,fill=black,inner sep=0pt, minimum width=4pt] (20) at (9.5*\breadthDist,\heightDist) {};
	\node[draw,circle,fill=black,inner sep=0pt, minimum width=4pt] (21) at (10.5*\breadthDist,\heightDist) {};
	\node[draw,circle,fill=black,inner sep=0pt, minimum width=4pt] (22) at (11.5*\breadthDist,\heightDist) {};
	\node[draw,circle,fill=black,inner sep=0pt, minimum width=4pt] (23) at (12.5*\breadthDist,\heightDist) {};
	\node[draw,circle,fill=black,inner sep=0pt, minimum width=4pt] (24) at (13.5*\breadthDist,\heightDist) {};
	\node[draw,circle,fill=black,inner sep=0pt, minimum width=4pt] (25) at (14.5*\breadthDist,\heightDist) {};
	\node[draw,circle,fill=black,inner sep=0pt, minimum width=4pt] (26) at (15.5*\breadthDist,\heightDist) {};
	\node[draw,circle,fill=black,inner sep=0pt, minimum width=4pt] (27) at (16.5*\breadthDist,\heightDist) {};
	\node[draw,circle,fill=black,inner sep=0pt, minimum width=4pt] (28) at (17.5*\breadthDist,\heightDist) {};
	\node[draw,circle,fill=black,inner sep=0pt, minimum width=4pt] (29) at (18.5*\breadthDist,\heightDist) {};
	\node[draw,circle,fill=black,inner sep=0pt, minimum width=4pt] (30) at (19.5*\breadthDist,\heightDist) {};
	\node[draw,circle,fill=black,inner sep=0pt, minimum width=4pt] (31) at (20.5*\breadthDist,\heightDist) {};		
	
	\node[draw,circle,fill=black,inner sep=0pt, minimum width=4pt] (15') at (-10.5*\breadthDist,0) {};
	\node[draw,circle,fill=black,inner sep=0pt, minimum width=4pt] (16') at (-9.5*\breadthDist,0) {};
	\node[draw,circle,fill=black,inner sep=0pt, minimum width=4pt] (17') at (-8.5*\breadthDist,0) {};
	\node[draw,circle,fill=black,inner sep=0pt, minimum width=4pt] (18') at (-7.5*\breadthDist,0) {};
	\node[draw,circle,fill=black,inner sep=0pt, minimum width=4pt] (19') at (-6.5*\breadthDist,0) {};
	\node[draw,circle,fill=black,inner sep=0pt, minimum width=4pt] (20') at (-5.5*\breadthDist,0) {};
	\node[draw,circle,fill=black,inner sep=0pt, minimum width=4pt] (21') at (-4.5*\breadthDist,0) {};
	\node[draw,circle,fill=black,inner sep=0pt, minimum width=4pt] (22') at (-3.5*\breadthDist,0) {};
	\node[draw,circle,fill=black,inner sep=0pt, minimum width=4pt] (23') at (-2.5*\breadthDist,0) {};
	\node[draw,circle,fill=black,inner sep=0pt, minimum width=4pt] (24') at (-1.5*\breadthDist,0) {};
	\node[draw,circle,fill=black,inner sep=0pt, minimum width=4pt] (25') at (-0.5*\breadthDist,0) {};
	\node[draw,circle,fill=black,inner sep=0pt, minimum width=4pt] (26') at (0.5*\breadthDist,0) {};
	\node[draw,circle,fill=black,inner sep=0pt, minimum width=4pt] (27') at (1.5*\breadthDist,0) {};
	\node[draw,circle,fill=black,inner sep=0pt, minimum width=4pt] (28') at (2.5*\breadthDist,0) {};
	\node[draw,circle,fill=black,inner sep=0pt, minimum width=4pt] (29') at (3.5*\breadthDist,0) {};
	\node[draw,circle,fill=black,inner sep=0pt, minimum width=4pt] (30') at (4.5*\breadthDist,0) {};
	\node[draw,circle,fill=black,inner sep=0pt, minimum width=4pt] (31') at (4.5*\breadthDist,-\breadthDist) {};
	\node[draw,circle,fill=black,inner sep=0pt, minimum width=4pt] (1') at (-5.5*\breadthDist+0.25*3*\breadthDist,0.75*\heightDist) {};
	\node[draw,circle,fill=black,inner sep=0pt, minimum width=4pt] (2') at (-5.5*\breadthDist+0.5*3*\breadthDist,0.5*\heightDist) {};
	\node[draw,circle,fill=black,inner sep=0pt, minimum width=4pt] (3') at (-5.5*\breadthDist+0.75*3*\breadthDist,0.25*\heightDist) {};
	\node[draw,circle,fill=black,inner sep=0pt, minimum width=4pt] (4') at (-1.5*\breadthDist+0.25*3*\breadthDist,0.25*\heightDist) {};
	\node[draw,circle,fill=black,inner sep=0pt, minimum width=4pt] (5') at (-1.5*\breadthDist+0.5*3*\breadthDist,0.5*\heightDist) {};
	\node[draw,circle,fill=black,inner sep=0pt, minimum width=4pt] (6') at (-1.5*\breadthDist+0.75*3*\breadthDist,0.75*\heightDist) {};
	
	\node[draw,circle,fill=black,inner sep=0pt, minimum width=4pt] (1'') at (6.5*\breadthDist,-\breadthDist) {};
	\node[draw,circle,fill=black,inner sep=0pt, minimum width=4pt] (2'') at (6.5*\breadthDist,0) {};
	\node[draw,circle,fill=black,inner sep=0pt, minimum width=4pt] (3'') at (7.5*\breadthDist,0) {};
	\node[draw,circle,fill=black,inner sep=0pt, minimum width=4pt] (4'') at (8.5*\breadthDist,0) {};
	\node[draw,circle,fill=black,inner sep=0pt, minimum width=4pt] (5'') at (9.5*\breadthDist,0) {};
	\node[draw,circle,fill=black,inner sep=0pt, minimum width=4pt] (6'') at (10.5*\breadthDist,0) {};
	\node[draw,circle,fill=black,inner sep=0pt, minimum width=4pt] (7'') at (11.5*\breadthDist,0) {};
	\node[draw,circle,fill=black,inner sep=0pt, minimum width=4pt] (8'') at (12.5*\breadthDist,0) {};
	\node[draw,circle,fill=black,inner sep=0pt, minimum width=4pt] (9'') at (13.5*\breadthDist,0) {};
	\node[draw,circle,fill=black,inner sep=0pt, minimum width=4pt] (10'') at (14.5*\breadthDist,0) {};
	\node[draw,circle,fill=black,inner sep=0pt, minimum width=4pt] (11'') at (15.5*\breadthDist,0) {};
	\node[draw,circle,fill=black,inner sep=0pt, minimum width=4pt] (12'') at (16.5*\breadthDist,0) {};
	\node[draw,circle,fill=black,inner sep=0pt, minimum width=4pt] (13'') at (17.5*\breadthDist,0) {};
	\node[draw,circle,fill=black,inner sep=0pt, minimum width=4pt] (14'') at (18.5*\breadthDist,0) {};
	\node[draw,circle,fill=black,inner sep=0pt, minimum width=4pt] (15'') at (19.5*\breadthDist,0) {};
	\node[draw,circle,fill=black,inner sep=0pt, minimum width=4pt] (16'') at (20.5*\breadthDist,0) {};
	\node[draw,circle,fill=black,inner sep=0pt, minimum width=4pt] (17'') at (21.5*\breadthDist,0) {};
	\node[draw,circle,fill=black,inner sep=0pt, minimum width=4pt] (1''') at (9.5*\breadthDist+0.25*3*\breadthDist,0.25*\heightDist) {};
	\node[draw,circle,fill=black,inner sep=0pt, minimum width=4pt] (2''') at (9.5*\breadthDist+0.5*3*\breadthDist,0.5*\heightDist) {};
	\node[draw,circle,fill=black,inner sep=0pt, minimum width=4pt] (3''') at (9.5*\breadthDist+0.75*3*\breadthDist,0.75*\heightDist) {};
	\node[draw,circle,fill=black,inner sep=0pt, minimum width=4pt] (4''') at (13.5*\breadthDist+0.25*3*\breadthDist,0.75*\heightDist) {};
	\node[draw,circle,fill=black,inner sep=0pt, minimum width=4pt] (5''') at (13.5*\breadthDist+0.5*3*\breadthDist,0.5*\heightDist) {};
	\node[draw,circle,fill=black,inner sep=0pt, minimum width=4pt] (6''') at (13.5*\breadthDist+0.75*3*\breadthDist,0.25*\heightDist) {};
	
	\path[C1,ns1,dashed] (-11*\breadthDist,\heightDist) edge (1);
	\path[C1,ns1]          (1)  edge  (2);
	\path[C1,ns1]          (2)  edge  (3);
	\path[C1,ns1]          (3)  edge  (4);
	\path[C1,ns1]          (4)  edge  (5);
	\path[C1,ns1]          (5)  edge  (6);
	\path[C1,ns1]          (6)  edge  (7);
	\path[C1,ns1]          (7)  edge  (8);
	\path[C1,ns1]          (8)  edge  (9);
	\path[C1,ns1]          (9)  edge  (10);
	\path[C1,ns1]          (10)  edge  (11);
	\path[C1,ns1]          (11)  edge  (12);
	\path[C1,ns1]          (12)  edge  (13);
	\path[C1,ns1]          (13)  edge  (14);
	\path[C1,ns1]          (14)  edge  (15);
	\path[C1,ns1]          (15)  edge  (16);
	\path[C1,ns1]          (16)  edge  (17);
	\path[C1,ns1]          (17)  edge  (18);
	\path[C1,ns1]          (18)  edge  (19);
	\path[C1,ns1]          (19)  edge  (20);
	\path[C1,ns1]          (20)  edge  (21);
	\path[C1,ns1]          (21)  edge  (22);
	\path[C1,ns1]          (22)  edge  (23);
	\path[C1,ns1]          (23)  edge  (24);
	\path[C1,ns1]          (24)  edge  (25);
	\path[C1,ns1]          (25)  edge  (26);
	\path[C1,ns1]          (26)  edge  (27);
	\path[C1,ns1]          (27)  edge  (28);
	\path[C1,ns1]          (28)  edge  (29);
	\path[C1,ns1]          (29)  edge  (30);
	\path[C1,ns1]          (30)  edge  (31);
	\path[C1,ns1,dashed]	(31) edge (22*\breadthDist,\heightDist);
	\path[C1,ns1,dashed]	(-11.3*\breadthDist,0) edge (15');
	\path[C1,ns1]          (15')  edge  (16');
	\path[C1,ns1]          (16')  edge  (17');
	\path[C1,ns1]          (17')  edge  (18');
	\path[C1,ns1]          (18')  edge  (19');
	\path[C1,ns1]          (19')  edge  (20');
	\path[C1,ns1]          (20')  edge  (21');
	\path[C1,ns1]          (21')  edge  (22');
	\path[C1,ns1]          (22')  edge  (23');
	\path[C1,ns1]          (23')  edge  (24');
	\path[C1,ns1]          (24')  edge  (25');
	\path[C1,ns1]          (25')  edge  (26');
	\path[C1,ns1]          (26')  edge  (27');
	\path[C1,ns1]          (27')  edge  (28');
	\path[C1,ns1]          (28')  edge  (29');
	\path[C1,ns1]          (29')  edge  (30');
	\path[C1,ns1]          (30')  edge  (31');
	\path[C1,ns1,dashed]	(4.5*\breadthDist,-\breadthDist-0.5) edge (31');
	\path[C1,ns1,dashed]	(6.5*\breadthDist,-\breadthDist-0.5) edge (1'');
	\path[C1,ns1]          (1'')  edge  (2'');
	\path[C1,ns1]          (2'')  edge  (3'');
	\path[C1,ns1]          (3'')  edge  (4'');
	\path[C1,ns1]          (4'')  edge  (5'');
	\path[C1,ns1]          (5'')  edge  (6'');
	\path[C1,ns1]          (6'')  edge  (7'');
	\path[C1,ns1]          (7'')  edge  (8'');
	\path[C1,ns1]          (8'')  edge  (9'');
	\path[C1,ns1]          (9'')  edge  (10'');
	\path[C1,ns1]          (10'')  edge  (11'');
	\path[C1,ns1]          (11'')  edge  (12'');
	\path[C1,ns1]          (12'')  edge  (13'');
	\path[C1,ns1]          (13'')  edge  (14'');
	\path[C1,ns1]          (14'')  edge  (15'');
	\path[C1,ns1]          (15'')  edge  (16'');
	\path[C1,ns1]          (16'')  edge  (17'');
	\path[C1,ns1,dashed]	(17'') edge (22.3*\breadthDist,0);
	\path[C1,ns1]          (3)  edge  (23');
	\path[C1,ns1]          (18')  edge  (8);
	\path[C1,ns1]          (9)  edge  (29');
	\path[C1,ns1]          (24')  edge  (14);
	\path[C1,ns1]          (5)  edge  (1');
	\path[C1,ns1]          (1')  edge  (2');
	\path[C1,ns1]          (2')  edge  (3');
	\path[C1,ns1]          (3')  edge  (23');
	\path[C1,ns1]          (24')  edge  (4');
	\path[C1,ns1]          (4')  edge  (5');
	\path[C1,ns1]          (5')  edge  (6');
	\path[C1,ns1]          (6')  edge  (12);
	\path[C1,ns1]          (3'')  edge  (23);
	\path[C1,ns1]          (18)  edge  (8'');
	\path[C1,ns1]          (9'')  edge  (29);
	\path[C1,ns1]          (24)  edge  (14'');
	\path[C1,ns1]          (5'')  edge  (1''');
	\path[C1,ns1]          (1''')  edge  (2''');
	\path[C1,ns1]          (2''')  edge  (3''');
	\path[C1,ns1]          (3''')  edge  (23);
	\path[C1,ns1]          (24)  edge  (4''');
	\path[C1,ns1]          (4''')  edge  (5''');
	\path[C1,ns1]          (5''')  edge  (6''');
	\path[C1,ns1]          (6''')  edge  (12'');
	\path[C1,ns1]          (2)  edge   [bend left] (5);
	\path[C1,ns1]          (6)  edge   [bend left] (11);
	\path[C1,ns1]          (12)  edge   [bend left] (17);
	\path[C1,ns1]          (15)  edge   [bend left] (20);
	\path[C1,ns1]          (21)  edge   [bend left] (26);
	\path[C1,ns1]          (27)  edge   [bend left] (30);
	\path[C1,ns1]          (-11*\breadthDist,-0.3*\heightShift)  edge   [bend right] (17');
	\path[C1,ns1]          (15')  edge   [bend right] (20');
	\path[C1,ns1]          (21')  edge   [bend right] (26');
	\path[C1,ns1]          (27')  edge   [bend right] (30');
	\path[C1,ns1]          (2'')  edge   [bend right] (5'');
	\path[C1,ns1]          (6'')  edge   [bend right] (11'');
	\path[C1,ns1]          (12'')  edge   [bend right] (17'');
	\path[C1,ns1]          (15'')  edge   [bend right] (22*\breadthDist,-0.3*\heightShift);
	\node[circle,inner sep=0pt] [shift={(-7.5*\breadthDist,\heightDist)}](3n1) at (317:\Rad) {};
	\node[circle,inner sep=0pt] [shift={(-7.5*\breadthDist,\heightDist)}](3nn1) at (317:\Radius) {};
	\path[C1,ns1]          (3)  edge  (3n1);
	\path[C1,ns1,densely dotted]          (3n1)  edge  (3nn1);
	\node[circle,inner sep=0pt] [shift={(-7.5*\breadthDist,\heightDist)}](3n2) at (312:\Rad) {};
	\node[circle,inner sep=0pt] [shift={(-7.5*\breadthDist,\heightDist)}](3nn2) at (312:\Radius) {};
	\path[C1,ns1]          (3)  edge  (3n2);
	\path[C1,ns1,densely dotted]          (3n2)  edge  (3nn2);
	\node[circle,inner sep=0pt] [shift={(-7.5*\breadthDist,\heightDist)}](3n3) at (307:\Rad) {};
	\node[circle,inner sep=0pt] [shift={(-7.5*\breadthDist,\heightDist)}](3nn3) at (307:\Radius) {};
	\path[C1,ns1]          (3)  edge  (3n3);
	\path[C1,ns1,densely dotted]          (3n3)  edge  (3nn3);
	\node[circle,inner sep=0pt] [shift={(-1.5*\breadthDist,\heightDist)}](9n1) at (317:\Rad) {};
	\node[circle,inner sep=0pt] [shift={(-1.5*\breadthDist,\heightDist)}](9nn1) at (317:\Radius) {};
	\path[C1,ns1]          (9)  edge  (9n1);
	\path[C1,ns1,densely dotted]          (9n1)  edge  (9nn1);
	\node[circle,inner sep=0pt] [shift={(-1.5*\breadthDist,\heightDist)}](9n2) at (312:\Rad) {};
	\node[circle,inner sep=0pt] [shift={(-1.5*\breadthDist,\heightDist)}](9nn2) at (312:\Radius) {};
	\path[C1,ns1]          (9)  edge  (9n2);
	\path[C1,ns1,densely dotted]          (9n2)  edge  (9nn2);
	\node[circle,inner sep=0pt] [shift={(-1.5*\breadthDist,\heightDist)}](9n3) at (307:\Rad) {};
	\node[circle,inner sep=0pt] [shift={(-1.5*\breadthDist,\heightDist)}](9nn3) at (307:\Radius) {};
	\path[C1,ns1]          (9)  edge  (9n3);
	\path[C1,ns1,densely dotted]          (9n3)  edge  (9nn3);
	\node[circle,inner sep=0pt] [shift={(7.5*\breadthDist,\heightDist)}](18n1) at (317:\Rad) {};
	\node[circle,inner sep=0pt] [shift={(7.5*\breadthDist,\heightDist)}](18nn1) at (317:\Radius) {};
	\path[C1,ns1]          (18)  edge  (18n1);
	\path[C1,ns1,densely dotted]          (18n1)  edge  (18nn1);
	\node[circle,inner sep=0pt] [shift={(7.5*\breadthDist,\heightDist)}](18n2) at (312:\Rad) {};
	\node[circle,inner sep=0pt] [shift={(7.5*\breadthDist,\heightDist)}](18nn2) at (312:\Radius) {};
	\path[C1,ns1]          (18)  edge  (18n2);
	\path[C1,ns1,densely dotted]          (18n2)  edge  (18nn2);
	\node[circle,inner sep=0pt] [shift={(7.5*\breadthDist,\heightDist)}](18n3) at (307:\Rad) {};
	\node[circle,inner sep=0pt] [shift={(7.5*\breadthDist,\heightDist)}](18nn3) at (307:\Radius) {};
	\path[C1,ns1]          (18)  edge  (18n3);
	\path[C1,ns1,densely dotted]          (18n3)  edge  (18nn3);
	\node[circle,inner sep=0pt] [shift={(-2.5*\breadthDist,\heightDist)}](8n1) at (233:\Rad) {};
	\node[circle,inner sep=0pt] [shift={(-2.5*\breadthDist,\heightDist)}](8nn1) at (233:\Radius) {};
	\path[C1,ns1]          (8)  edge  (8n1);
	\path[C1,ns1,densely dotted]          (8n1)  edge  (8nn1);
	\node[circle,inner sep=0pt] [shift={(-2.5*\breadthDist,\heightDist)}](8n2) at (228:\Rad) {};
	\node[circle,inner sep=0pt] [shift={(-2.5*\breadthDist,\heightDist)}](8nn2) at (228:\Radius) {};
	\path[C1,ns1]          (8)  edge  (8n2);
	\path[C1,ns1,densely dotted]          (8n2)  edge  (8nn2);
	\node[circle,inner sep=0pt] [shift={(-2.5*\breadthDist,\heightDist)}](8n3) at (223:\Rad) {};
	\node[circle,inner sep=0pt] [shift={(-2.5*\breadthDist,\heightDist)}](8nn3) at (223:\Radius) {};
	\path[C1,ns1]          (8)  edge  (8n3);
	\path[C1,ns1,densely dotted]          (8n3)  edge  (8nn3);
	\node[circle,inner sep=0pt] [shift={(3.5*\breadthDist,\heightDist)}](14n1) at (233:\Rad) {};
	\node[circle,inner sep=0pt] [shift={(3.5*\breadthDist,\heightDist)}](14nn1) at (233:\Radius) {};
	\path[C1,ns1]          (14)  edge  (14n1);
	\path[C1,ns1,densely dotted]          (14n1)  edge  (14nn1);
	\node[circle,inner sep=0pt] [shift={(3.5*\breadthDist,\heightDist)}](14n2) at (228:\Rad) {};
	\node[circle,inner sep=0pt] [shift={(3.5*\breadthDist,\heightDist)}](14nn2) at (228:\Radius) {};
	\path[C1,ns1]          (14)  edge  (14n2);
	\path[C1,ns1,densely dotted]          (14n2)  edge  (14nn2);
	\node[circle,inner sep=0pt] [shift={(3.5*\breadthDist,\heightDist)}](14n3) at (223:\Rad) {};
	\node[circle,inner sep=0pt] [shift={(3.5*\breadthDist,\heightDist)}](14nn3) at (223:\Radius) {};
	\path[C1,ns1]          (14)  edge  (14n3);
	\path[C1,ns1,densely dotted]          (14n3)  edge  (14nn3);
	\node[circle,inner sep=0pt] [shift={(18.5*\breadthDist,\heightDist)}](29n1) at (233:\Rad) {};
	\node[circle,inner sep=0pt] [shift={(18.5*\breadthDist,\heightDist)}](29nn1) at (233:\Radius) {};
	\path[C1,ns1]          (29)  edge  (29n1);
	\path[C1,ns1,densely dotted]          (29n1)  edge  (29nn1);
	\node[circle,inner sep=0pt] [shift={(18.5*\breadthDist,\heightDist)}](29n2) at (228:\Rad) {};
	\node[circle,inner sep=0pt] [shift={(18.5*\breadthDist,\heightDist)}](29nn2) at (228:\Radius) {};
	\path[C1,ns1]          (29)  edge  (29n2);
	\path[C1,ns1,densely dotted]          (29n2)  edge  (29nn2);
	\node[circle,inner sep=0pt] [shift={(18.5*\breadthDist,\heightDist)}](29n3) at (223:\Rad) {};
	\node[circle,inner sep=0pt] [shift={(18.5*\breadthDist,\heightDist)}](29nn3) at (223:\Radius) {};
	\path[C1,ns1]          (29)  edge  (29n3);
	\path[C1,ns1,densely dotted]          (29n3)  edge  (29nn3);
	\node[circle,inner sep=0pt] [shift={(12.5*\breadthDist,\heightDist)}](23n1) at (203:\Rad) {};
	\node[circle,inner sep=0pt] [shift={(12.5*\breadthDist,\heightDist)}](23nn1) at (203:\Radius) {};
	\path[C1,ns1]          (23)  edge  (23n1);
	\path[C1,ns1,densely dotted]          (23n1)  edge  (23nn1);
	\node[circle,inner sep=0pt] [shift={(12.5*\breadthDist,\heightDist)}](23n2) at (208:\Rad) {};
	\node[circle,inner sep=0pt] [shift={(12.5*\breadthDist,\heightDist)}](23nn2) at (208:\Radius) {};
	\path[C1,ns1]          (23)  edge  (23n2);
	\path[C1,ns1,densely dotted]          (23n2)  edge  (23nn2);
	\node[circle,inner sep=0pt] [shift={(12.5*\breadthDist,\heightDist)}](23n3) at (213:\Rad) {};
	\node[circle,inner sep=0pt] [shift={(12.5*\breadthDist,\heightDist)}](23nn3) at (213:\Radius) {};
	\path[C1,ns1]          (23)  edge  (23n3);
	\path[C1,ns1,densely dotted]          (23n3)  edge  (23nn3);
	\node[circle,inner sep=0pt] [shift={(13.5*\breadthDist,\heightDist)}](24n1) at (337:\Rad) {};
	\node[circle,inner sep=0pt] [shift={(13.5*\breadthDist,\heightDist)}](24nn1) at (337:\Radius) {};
	\path[C1,ns1]          (24)  edge  (24n1);
	\path[C1,ns1,densely dotted]          (24n1)  edge  (24nn1);
	\node[circle,inner sep=0pt] [shift={(13.5*\breadthDist,\heightDist)}](24n2) at (332:\Rad) {};
	\node[circle,inner sep=0pt] [shift={(13.5*\breadthDist,\heightDist)}](24nn2) at (332:\Radius) {};
	\path[C1,ns1]          (24)  edge  (24n2);
	\path[C1,ns1,densely dotted]          (24n2)  edge  (24nn2);
	\node[circle,inner sep=0pt] [shift={(13.5*\breadthDist,\heightDist)}](24n3) at (327:\Rad) {};
	\node[circle,inner sep=0pt] [shift={(13.5*\breadthDist,\heightDist)}](24nn3) at (327:\Radius) {};
	\path[C1,ns1]          (24)  edge  (24n3);
	\path[C1,ns1,densely dotted]          (24n3)  edge  (24nn3);
	\node[circle,inner sep=0pt] [shift={(-5.5*\breadthDist+0.25*3*\breadthDist,0.75*\heightDist)}](1'n1) at (98:\rad) {};
	\node[circle,inner sep=0pt] [shift={(-5.5*\breadthDist+0.25*3*\breadthDist,0.75*\heightDist)}](1'nn1) at (98:\radius) {};
	\path[C1,ns1]          (1')  edge  (1'n1);
	\path[C1,ns1,densely dotted]          (1'n1)  edge  (1'nn1);
	\node[circle,inner sep=0pt] [shift={(-5.5*\breadthDist+0.25*3*\breadthDist,0.75*\heightDist)}](1'n2) at (108:\rad) {};
	\node[circle,inner sep=0pt] [shift={(-5.5*\breadthDist+0.25*3*\breadthDist,0.75*\heightDist)}](1'nn2) at (108:\radius) {};
	\path[C1,ns1]          (1')  edge  (1'n2);
	\path[C1,ns1,densely dotted]          (1'n2)  edge  (1'nn2);
	\node[circle,inner sep=0pt] [shift={(-5.5*\breadthDist+0.25*3*\breadthDist,0.75*\heightDist)}](1'n3) at (118:\rad) {};
	\node[circle,inner sep=0pt] [shift={(-5.5*\breadthDist+0.25*3*\breadthDist,0.75*\heightDist)}](1'nn3) at (118:\radius) {};
	\path[C1,ns1]          (1')  edge  (1'n3);
	\path[C1,ns1,densely dotted]          (1'n3)  edge  (1'nn3);
	\node[circle,inner sep=0pt] [shift={(-1.5*\breadthDist+0.75*3*\breadthDist,0.75*\heightDist)}](6'n1) at (82:\rad) {};
	\node[circle,inner sep=0pt] [shift={(-1.5*\breadthDist+0.75*3*\breadthDist,0.75*\heightDist)}](6'nn1) at (82:\radius) {};
	\path[C1,ns1]          (6')  edge  (6'n1);
	\path[C1,ns1,densely dotted]          (6'n1)  edge  (6'nn1);
	\node[circle,inner sep=0pt] [shift={(-1.5*\breadthDist+0.75*3*\breadthDist,0.75*\heightDist)}](6'n2) at (72:\rad) {};
	\node[circle,inner sep=0pt] [shift={(-1.5*\breadthDist+0.75*3*\breadthDist,0.75*\heightDist)}](6'nn2) at (72:\radius) {};
	\path[C1,ns1]          (6')  edge  (6'n2);
	\path[C1,ns1,densely dotted]          (6'n2)  edge  (6'nn2);
	\node[circle,inner sep=0pt] [shift={(-1.5*\breadthDist+0.75*3*\breadthDist,0.75*\heightDist)}](6'n3) at (62:\rad) {};
	\node[circle,inner sep=0pt] [shift={(-1.5*\breadthDist+0.75*3*\breadthDist,0.75*\heightDist)}](6'nn3) at (62:\radius) {};
	\path[C1,ns1]          (6')  edge  (6'n3);
	\path[C1,ns1,densely dotted]          (6'n3)  edge  (6'nn3);
	\node[circle,inner sep=0pt] [shift={(9.5*\breadthDist+0.75*3*\breadthDist,0.75*\heightDist)}](3'''n1) at (24:\rad) {};
	\node[circle,inner sep=0pt] [shift={(9.5*\breadthDist+0.75*3*\breadthDist,0.75*\heightDist)}](3'''nn1) at (24:\radius) {};
	\path[C1,ns1]          (3''')  edge  (3'''n1);
	\path[C1,ns1,densely dotted]          (3'''n1)  edge  (3'''nn1);
	\node[circle,inner sep=0pt] [shift={(9.5*\breadthDist+0.75*3*\breadthDist,0.75*\heightDist)}](3'''n2) at (34:\rad) {};
	\node[circle,inner sep=0pt] [shift={(9.5*\breadthDist+0.75*3*\breadthDist,0.75*\heightDist)}](3'''nn2) at (34:\radius) {};
	\path[C1,ns1]          (3''')  edge  (3'''n2);
	\path[C1,ns1,densely dotted]          (3'''n2)  edge  (3'''nn2);
	\node[circle,inner sep=0pt] [shift={(9.5*\breadthDist+0.75*3*\breadthDist,0.75*\heightDist)}](3'''n3) at (44:\rad) {};
	\node[circle,inner sep=0pt] [shift={(9.5*\breadthDist+0.75*3*\breadthDist,0.75*\heightDist)}](3'''nn3) at (44:\radius) {};
	\path[C1,ns1]          (3''')  edge  (3'''n3);
	\path[C1,ns1,densely dotted]          (3'''n3)  edge  (3'''nn3);
	\node[circle,inner sep=0pt] [shift={(13.5*\breadthDist+0.25*3*\breadthDist,0.75*\heightDist)}](4'''n1) at (156:\rad) {};
	\node[circle,inner sep=0pt] [shift={(13.5*\breadthDist+0.25*3*\breadthDist,0.75*\heightDist)}](4'''nn1) at (156:\radius) {};
	\path[C1,ns1]          (4''')  edge  (4'''n1);
	\path[C1,ns1,densely dotted]          (4'''n1)  edge  (4'''nn1);
	\node[circle,inner sep=0pt] [shift={(13.5*\breadthDist+0.25*3*\breadthDist,0.75*\heightDist)}](4'''n2) at (146:\rad) {};
	\node[circle,inner sep=0pt] [shift={(13.5*\breadthDist+0.25*3*\breadthDist,0.75*\heightDist)}](4'''nn2) at (146:\radius) {};
	\path[C1,ns1]          (4''')  edge  (4'''n2);
	\path[C1,ns1,densely dotted]          (4'''n2)  edge  (4'''nn2);
	\node[circle,inner sep=0pt] [shift={(13.5*\breadthDist+0.25*3*\breadthDist,0.75*\heightDist)}](4'''n3) at (136:\rad) {};
	\node[circle,inner sep=0pt] [shift={(13.5*\breadthDist+0.25*3*\breadthDist,0.75*\heightDist)}](4'''nn3) at (136:\radius) {};
	\path[C1,ns1]          (4''')  edge  (4'''n3);
	\path[C1,ns1,densely dotted]          (4'''n3)  edge  (4'''nn3);
	\node[circle,inner sep=0pt] [shift={(-7.5*\breadthDist,0)}](18'n1) at (43:\Rad) {};
	\node[circle,inner sep=0pt] [shift={(-7.5*\breadthDist,0)}](18'nn1) at (43:\Radius) {};
	\path[C1,ns1]          (18')  edge  (18'n1);
	\path[C1,ns1,densely dotted]          (18'n1)  edge  (18'nn1);
	\node[circle,inner sep=0pt] [shift={(-7.5*\breadthDist,0)}](18'n2) at (48:\Rad) {};
	\node[circle,inner sep=0pt] [shift={(-7.5*\breadthDist,0)}](18'nn2) at (48:\Radius) {};
	\path[C1,ns1]          (18')  edge  (18'n2);
	\path[C1,ns1,densely dotted]          (18'n2)  edge  (18'nn2);
	\node[circle,inner sep=0pt] [shift={(-7.5*\breadthDist,0)}](18'n3) at (53:\Rad) {};
	\node[circle,inner sep=0pt] [shift={(-7.5*\breadthDist,0)}](18'nn3) at (53:\Radius) {};
	\path[C1,ns1]          (18')  edge  (18'n3);
	\path[C1,ns1,densely dotted]          (18'n3)  edge  (18'nn3);
	\node[circle,inner sep=0pt] [shift={(-2.5*\breadthDist,0)}](23'n1) at (157:\Rad) {};
	\node[circle,inner sep=0pt] [shift={(-2.5*\breadthDist,0)}](23'nn1) at (157:\Radius) {};
	\path[C1,ns1]          (23')  edge  (23'n1);
	\path[C1,ns1,densely dotted]          (23'n1)  edge  (23'nn1);
	\node[circle,inner sep=0pt] [shift={(-2.5*\breadthDist,0)}](23'n2) at (152:\Rad) {};
	\node[circle,inner sep=0pt] [shift={(-2.5*\breadthDist,0)}](23'nn2) at (152:\Radius) {};
	\path[C1,ns1]          (23')  edge  (23'n2);
	\path[C1,ns1,densely dotted]          (23'n2)  edge  (23'nn2);
	\node[circle,inner sep=0pt] [shift={(-2.5*\breadthDist,0)}](23'n3) at (147:\Rad) {};
	\node[circle,inner sep=0pt] [shift={(-2.5*\breadthDist,0)}](23'nn3) at (147:\Radius) {};
	\path[C1,ns1]          (23')  edge  (23'n3);
	\path[C1,ns1,densely dotted]          (23'n3)  edge  (23'nn3);
	\node[circle,inner sep=0pt] [shift={(-1.5*\breadthDist,0)}](24'n1) at (23:\Rad) {};
	\node[circle,inner sep=0pt] [shift={(-1.5*\breadthDist,0)}](24'nn1) at (23:\Radius) {};
	\path[C1,ns1]          (24')  edge  (24'n1);
	\path[C1,ns1,densely dotted]          (24'n1)  edge  (24'nn1);
	\node[circle,inner sep=0pt] [shift={(-1.5*\breadthDist,0)}](24'n2) at (28:\Rad) {};
	\node[circle,inner sep=0pt] [shift={(-1.5*\breadthDist,0)}](24'nn2) at (28:\Radius) {};
	\path[C1,ns1]          (24')  edge  (24'n2);
	\path[C1,ns1,densely dotted]          (24'n2)  edge  (24'nn2);
	\node[circle,inner sep=0pt] [shift={(-1.5*\breadthDist,0)}](24'n3) at (33:\Rad) {};
	\node[circle,inner sep=0pt] [shift={(-1.5*\breadthDist,0)}](24'nn3) at (33:\Radius) {};
	\path[C1,ns1]          (24')  edge  (24'n3);
	\path[C1,ns1,densely dotted]          (24'n3)  edge  (24'nn3);
	\node[circle,inner sep=0pt] [shift={(3.5*\breadthDist,0)}](29'n1) at (137:\Rad) {};
	\node[circle,inner sep=0pt] [shift={(3.5*\breadthDist,0)}](29'nn1) at (137:\Radius) {};
	\path[C1,ns1]          (29')  edge  (29'n1);
	\path[C1,ns1,densely dotted]          (29'n1)  edge  (29'nn1);
	\node[circle,inner sep=0pt] [shift={(3.5*\breadthDist,0)}](29'n2) at (132:\Rad) {};
	\node[circle,inner sep=0pt] [shift={(3.5*\breadthDist,0)}](29'nn2) at (132:\Radius) {};
	\path[C1,ns1]          (29')  edge  (29'n2);
	\path[C1,ns1,densely dotted]          (29'n2)  edge  (29'nn2);
	\node[circle,inner sep=0pt] [shift={(3.5*\breadthDist,0)}](29'n3) at (127:\Rad) {};
	\node[circle,inner sep=0pt] [shift={(3.5*\breadthDist,0)}](29'nn3) at (127:\Radius) {};
	\path[C1,ns1]          (29')  edge  (29'n3);
	\path[C1,ns1,densely dotted]          (29'n3)  edge  (29'nn3);
	\node[circle,inner sep=0pt] [shift={(7.5*\breadthDist,0)}](3''n1) at (43:\Rad) {};
	\node[circle,inner sep=0pt] [shift={(7.5*\breadthDist,0)}](3''nn1) at (43:\Radius) {};
	\path[C1,ns1]          (3'')  edge  (3''n1);
	\path[C1,ns1,densely dotted]          (3''n1)  edge  (3''nn1);
	\node[circle,inner sep=0pt] [shift={(7.5*\breadthDist,0)}](3''n2) at (48:\Rad) {};
	\node[circle,inner sep=0pt] [shift={(7.5*\breadthDist,0)}](3''nn2) at (48:\Radius) {};
	\path[C1,ns1]          (3'')  edge  (3''n2);
	\path[C1,ns1,densely dotted]          (3''n2)  edge  (3''nn2);
	\node[circle,inner sep=0pt] [shift={(7.5*\breadthDist,0)}](3''n3) at (53:\Rad) {};
	\node[circle,inner sep=0pt] [shift={(7.5*\breadthDist,0)}](3''nn3) at (53:\Radius) {};
	\path[C1,ns1]          (3'')  edge  (3''n3);
	\path[C1,ns1,densely dotted]          (3''n3)  edge  (3''nn3);
	\node[circle,inner sep=0pt] [shift={(12.5*\breadthDist,0)}](8''n1) at (137:\Rad) {};
	\node[circle,inner sep=0pt] [shift={(12.5*\breadthDist,0)}](8''nn1) at (137:\Radius) {};
	\path[C1,ns1]          (8'')  edge  (8''n1);
	\path[C1,ns1,densely dotted]          (8''n1)  edge  (8''nn1);
	\node[circle,inner sep=0pt] [shift={(12.5*\breadthDist,0)}](8''n2) at (132:\Rad) {};
	\node[circle,inner sep=0pt] [shift={(12.5*\breadthDist,0)}](8''nn2) at (132:\Radius) {};
	\path[C1,ns1]          (8'')  edge  (8''n2);
	\path[C1,ns1,densely dotted]          (8''n2)  edge  (8''nn2);
	\node[circle,inner sep=0pt] [shift={(12.5*\breadthDist,0)}](8''n3) at (127:\Rad) {};
	\node[circle,inner sep=0pt] [shift={(12.5*\breadthDist,0)}](8''nn3) at (127:\Radius) {};
	\path[C1,ns1]          (8'')  edge  (8''n3);
	\path[C1,ns1,densely dotted]          (8''n3)  edge  (8''nn3);
	\node[circle,inner sep=0pt] [shift={(13.5*\breadthDist,0)}](9''n1) at (43:\Rad) {};
	\node[circle,inner sep=0pt] [shift={(13.5*\breadthDist,0)}](9''nn1) at (43:\Radius) {};
	\path[C1,ns1]          (9'')  edge  (9''n1);
	\path[C1,ns1,densely dotted]          (9''n1)  edge  (9''nn1);
	\node[circle,inner sep=0pt] [shift={(13.5*\breadthDist,0)}](9''n2) at (48:\Rad) {};
	\node[circle,inner sep=0pt] [shift={(13.5*\breadthDist,0)}](9''nn2) at (48:\Radius) {};
	\path[C1,ns1]          (9'')  edge  (9''n2);
	\path[C1,ns1,densely dotted]          (9''n2)  edge  (9''nn2);
	\node[circle,inner sep=0pt] [shift={(13.5*\breadthDist,0)}](9''n3) at (53:\Rad) {};
	\node[circle,inner sep=0pt] [shift={(13.5*\breadthDist,0)}](9''nn3) at (53:\Radius) {};
	\path[C1,ns1]          (9'')  edge  (9''n3);
	\path[C1,ns1,densely dotted]          (9''n3)  edge  (9''nn3);
	\node[circle,inner sep=0pt] [shift={(18.5*\breadthDist,0)}](14''n1) at (137:\Rad) {};
	\node[circle,inner sep=0pt] [shift={(18.5*\breadthDist,0)}](14''nn1) at (137:\Radius) {};
	\path[C1,ns1]          (14'')  edge  (14''n1);
	\path[C1,ns1,densely dotted]          (14''n1)  edge  (14''nn1);
	\node[circle,inner sep=0pt] [shift={(18.5*\breadthDist,0)}](14''n2) at (132:\Rad) {};
	\node[circle,inner sep=0pt] [shift={(18.5*\breadthDist,0)}](14''nn2) at (132:\Radius) {};
	\path[C1,ns1]          (14'')  edge  (14''n2);
	\path[C1,ns1,densely dotted]          (14''n2)  edge  (14''nn2);
	\node[circle,inner sep=0pt] [shift={(18.5*\breadthDist,0)}](14''n3) at (127:\Rad) {};
	\node[circle,inner sep=0pt] [shift={(18.5*\breadthDist,0)}](14''nn3) at (127:\Radius) {};
	\path[C1,ns1]          (14'')  edge  (14''n3);
	\path[C1,ns1,densely dotted]          (14''n3)  edge  (14''nn3);
	\node[circle,inner sep=0pt] [shift={(-5.5*\breadthDist+0.75*3*\breadthDist,0.25*\heightDist)}](3'n1) at (336:\rad) {};
	\node[circle,inner sep=0pt] [shift={(-5.5*\breadthDist+0.75*3*\breadthDist,0.25*\heightDist)}](3'nn1) at (336:\radius) {};
	\path[C1,ns1]          (3')  edge  (3'n1);
	\path[C1,ns1,densely dotted]          (3'n1)  edge  (3'nn1);
	\node[circle,inner sep=0pt] [shift={(-5.5*\breadthDist+0.75*3*\breadthDist,0.25*\heightDist)}](3'n2) at (326:\rad) {};
	\node[circle,inner sep=0pt] [shift={(-5.5*\breadthDist+0.75*3*\breadthDist,0.25*\heightDist)}](3'nn2) at (326:\radius) {};
	\path[C1,ns1]          (3')  edge  (3'n2);
	\path[C1,ns1,densely dotted]          (3'n2)  edge  (3'nn2);
	\node[circle,inner sep=0pt] [shift={(-5.5*\breadthDist+0.75*3*\breadthDist,0.25*\heightDist)}](3'n3) at (316:\rad) {};
	\node[circle,inner sep=0pt] [shift={(-5.5*\breadthDist+0.75*3*\breadthDist,0.25*\heightDist)}](3'nn3) at (316:\radius) {};
	\path[C1,ns1]          (3')  edge  (3'n3);
	\path[C1,ns1,densely dotted]          (3'n3)  edge  (3'nn3);
	\node[circle,inner sep=0pt] [shift={(-1.5*\breadthDist+0.25*3*\breadthDist,0.25*\heightDist)}](4'n1) at (204:\rad) {};
	\node[circle,inner sep=0pt] [shift={(-1.5*\breadthDist+0.25*3*\breadthDist,0.25*\heightDist)}](4'nn1) at (204:\radius) {};
	\path[C1,ns1]          (4')  edge  (4'n1);
	\path[C1,ns1,densely dotted]          (4'n1)  edge  (4'nn1);
	\node[circle,inner sep=0pt] [shift={(-1.5*\breadthDist+0.25*3*\breadthDist,0.25*\heightDist)}](4'n2) at (214:\rad) {};
	\node[circle,inner sep=0pt] [shift={(-1.5*\breadthDist+0.25*3*\breadthDist,0.25*\heightDist)}](4'nn2) at (214:\radius) {};
	\path[C1,ns1]          (4')  edge  (4'n2);
	\path[C1,ns1,densely dotted]          (4'n2)  edge  (4'nn2);
	\node[circle,inner sep=0pt] [shift={(-1.5*\breadthDist+0.25*3*\breadthDist,0.25*\heightDist)}](4'n3) at (224:\rad) {};
	\node[circle,inner sep=0pt] [shift={(-1.5*\breadthDist+0.25*3*\breadthDist,0.25*\heightDist)}](4'nn3) at (224:\radius) {};
	\path[C1,ns1]          (4')  edge  (4'n3);
	\path[C1,ns1,densely dotted]          (4'n3)  edge  (4'nn3);
	\node[circle,inner sep=0pt] [shift={(9.5*\breadthDist+0.25*3*\breadthDist,0.25*\heightDist)}](1'''n1) at (262:\rad) {};
	\node[circle,inner sep=0pt] [shift={(9.5*\breadthDist+0.25*3*\breadthDist,0.25*\heightDist)}](1'''nn1) at (262:\radius) {};
	\path[C1,ns1]          (1''')  edge  (1'''n1);
	\path[C1,ns1,densely dotted]          (1'''n1)  edge  (1'''nn1);
	\node[circle,inner sep=0pt] [shift={(9.5*\breadthDist+0.25*3*\breadthDist,0.25*\heightDist)}](1'''n2) at (252:\rad) {};
	\node[circle,inner sep=0pt] [shift={(9.5*\breadthDist+0.25*3*\breadthDist,0.25*\heightDist)}](1'''nn2) at (252:\radius) {};
	\path[C1,ns1]          (1''')  edge  (1'''n2);
	\path[C1,ns1,densely dotted]          (1'''n2)  edge  (1'''nn2);
	\node[circle,inner sep=0pt] [shift={(9.5*\breadthDist+0.25*3*\breadthDist,0.25*\heightDist)}](1'''n3) at (242:\rad) {};
	\node[circle,inner sep=0pt] [shift={(9.5*\breadthDist+0.25*3*\breadthDist,0.25*\heightDist)}](1'''nn3) at (242:\radius) {};
	\path[C1,ns1]          (1''')  edge  (1'''n3);
	\path[C1,ns1,densely dotted]          (1'''n3)  edge  (1'''nn3);
	\node[circle,inner sep=0pt] [shift={(13.5*\breadthDist+0.75*3*\breadthDist,0.25*\heightDist)}](6'''n1) at (278:\rad) {};
	\node[circle,inner sep=0pt] [shift={(13.5*\breadthDist+0.75*3*\breadthDist,0.25*\heightDist)}](6'''nn1) at (278:\radius) {};
	\path[C1,ns1]          (6''')  edge  (6'''n1);
	\path[C1,ns1,densely dotted]          (6'''n1)  edge  (6'''nn1);
	\node[circle,inner sep=0pt] [shift={(13.5*\breadthDist+0.75*3*\breadthDist,0.25*\heightDist)}](6'''n2) at (288:\rad) {};
	\node[circle,inner sep=0pt] [shift={(13.5*\breadthDist+0.75*3*\breadthDist,0.25*\heightDist)}](6'''nn2) at (288:\radius) {};
	\path[C1,ns1]          (6''')  edge  (6'''n2);
	\path[C1,ns1,densely dotted]          (6'''n2)  edge  (6'''nn2);
	\node[circle,inner sep=0pt] [shift={(13.5*\breadthDist+0.75*3*\breadthDist,0.25*\heightDist)}](6'''n3) at (298:\rad) {};
	\node[circle,inner sep=0pt] [shift={(13.5*\breadthDist+0.75*3*\breadthDist,0.25*\heightDist)}](6'''nn3) at (298:\radius) {};
	\path[C1,ns1]          (6''')  edge  (6'''n3);
	\path[C1,ns1,densely dotted]          (6'''n3)  edge  (6'''nn3);

	\node[minimum height=10pt,inner sep=0,font=\small,C1] at (-9.5*\breadthDist+0.2,\heightDist+0.3) {$a_1^{(v,w)}$};
	\node[minimum height=10pt,inner sep=0,font=\small,C1] at (-4.5*\breadthDist+0.2,\heightDist+0.3) {$a_{6}^{(v,w)}$};
	\node[minimum height=10pt,inner sep=0,font=\small,C1] at (1.5*\breadthDist+0.13,\heightDist+0.3) {$a_{12}^{(v,w)}$};
	\node[minimum height=10pt,inner sep=0,font=\small,C1] at (10.5*\breadthDist+0.13,\heightDist+0.3) {$a_{21}^{(v,w)}$};
	\node[minimum height=10pt,inner sep=0,font=\small,C1] at (16.5*\breadthDist+0.13,\heightDist+0.3) {$a_{27}^{(v,w)}$};
	\node[minimum height=10pt,inner sep=0,font=\small,C1] at (-5.5*\breadthDist,-0.3) {$a_{20}^{(u,v)}$};
	\node[minimum height=10pt,inner sep=0,font=\small,C1] at (0.5*\breadthDist,-0.3) {$a_{26}^{(u,v)}$};
	\node[minimum height=10pt,inner sep=0,font=\small,C1] at (9.5*\breadthDist,-0.3) {$a_{5}^{(w,x)}$};
	\node[minimum height=10pt,inner sep=0,font=\small,C1] at (15.5*\breadthDist,-0.3) {$a_{11}^{(w,x)}$};
	\node[minimum height=10pt,inner sep=0,font=\small,C1] at (-5.5*\breadthDist+0.5*3*\breadthDist-0.57,0.5*\heightDist+0.45) {$b_{1}^v$};
	\node[minimum height=10pt,inner sep=0,font=\small,C1] at (-5.5*\breadthDist+0.5*3*\breadthDist+0.3,0.5*\heightDist) {$b_{2}^v$};
	\node[minimum height=10pt,inner sep=0,font=\small,C1] at (-1.5*\breadthDist+0.5*3*\breadthDist-0.3,0.5*\heightDist) {$b_{5}^v$};
	\node[minimum height=10pt,inner sep=0,font=\small,C1] at (-1.5*\breadthDist+0.5*3*\breadthDist+0.57,0.5*\heightDist+0.45) {$b_{6}^v$};
	\node[minimum height=10pt,inner sep=0,font=\small,C1] at (9.5*\breadthDist+0.5*3*\breadthDist-0.55,0.5*\heightDist-0.49) {$b_{1}^w$};
	\node[minimum height=10pt,inner sep=0,font=\small,C1] at (9.5*\breadthDist+0.5*3*\breadthDist+0.3,0.5*\heightDist) {$b_{2}^w$};
	\node[minimum height=10pt,inner sep=0,font=\small,C1] at (13.5*\breadthDist+0.5*3*\breadthDist-0.3,0.5*\heightDist) {$b_{5}^w$};
	\node[minimum height=10pt,inner sep=0,font=\small,C1] at (13.5*\breadthDist+0.5*3*\breadthDist+0.58,0.5*\heightDist-0.49) {$b_{6}^w$};

	\end{tikzpicture}
	\caption{Close-up of $G_\mathcal{E}$ with vertices of high degree ($d+1$, $d+2$ or $d+3$) indicated by `fans'. }\label{fig:close-up}
\end{figure}
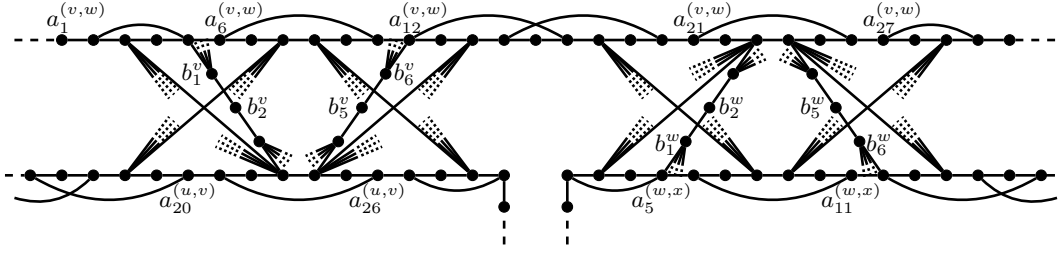 
\begin{proof}
	To prove (\ref{order1}) let us observe that $a_1^e$ and $a_4^e$ have degree $2$ in $G$, as $G_\mathcal{E}[N_{G_\mathcal{E}}(S_v)]\cong G[N_G(S_v)]$ and $a_1^e$ and $a_4^e$ have degree $2$ in $G_\mathcal{E}$. Hence $(a_1^e,a_2^e)$ and $(a_3^e,a_4^e,a_5^e)$ have to be subpaths of $C$ as in Remark~\ref{rem:vertexOfDegreeTwo}. Since $a_2^e$ has exactly three neighbours $a_1^e,a_3^e$ and $a_5^e$ in $G_\mathcal{E}$ and $G_\mathcal{E}[N_{G_\mathcal{E}}(S_v)]\cong G[N_G(S_v)]$ we get that either $(a_1^e,\dots,a_5^e)$ is a subpath of $C$ or $(a_1^e,a_2^e,a_5^e,a_4^e,a_3^e)$ is a subpath of $C$. 
	Property (\ref{order2}) follows with a similar argumentation. 
	
	For (\ref{order3}) let us assume that neither $(a_{12}^e,\dots, a_{17}^e)$ nor $(a_{14}^e,a_{13}^e,a_{12}^e,a_{17}^e,a_{16}^e,a_{15}^e)$ appear in $C$ as a subpath. 
	Since both $a_{13}^e$ and $a_{16}^e$ have degree $2$ in $G$, we know that $(a_{12}^e,a_{13}^e,a_{14}^e)$ and $(a_{15}^e,a_{16}^e,a_{17}^e)$ are subpaths of $C$. 
	Hence neither $(a_{14}^e,a_{15}^e)$ nor $(a_{12}^e,a_{17}^e)$ are subpaths of $C$. Since both $a_{15}^e$ and $a_{17}^e$ have degree $3$ in $G$, this implies that $(a_{20}^e,a_{15}^e,a_{16}^e,a_{17}^e,a_{18}^e)$ is a subpath of $C$. 
	Since $a_{19}^e$ has degree $2$, then $(a_{20}^e,a_{15}^e,a_{16}^e,a_{17}^e,a_{18}^e,a_{19}^e,a_{20}^e)$ has to be a subpath of $C$. Since this is a cycle, $C$ must be equal to $(a_{20}^e,a_{15}^e,a_{16}^e,a_{17}^e,a_{18}^e,a_{19}^e,a_{20}^e)$ which contradicts the assumption that $S_v$ is contained in $C$.
	A symmetric argument shows that either $(a_{15}^e,\dots,a_{20}^e)$ or $(a_{17}^e,a_{16}^e,a_{15}^e,a_{20}^e,a_{19}^e,a_{18}^e)$ has to be a subpath of $C$, proving (\ref{order3}).
	
	We will prove  (\ref{order4}) and (\ref{order5}) simultaneously using a counting argument. Let us first observe that for every edge $e\in E(\vec{\mathcal{E}})$ incident to $v$ we know that $(a_6^{e},a_7^{e},a_8^{e})$, $(a_9^{e},a_{10}^{e},a_{11}^{e})$, $(a_{21}^{e},a_{22}^{e},a_{23}^{e})$ and $(a_{24}^{e},a_{25}^{e},a_{26}^{e})$ are subpaths of $C$, because $a_7^{e}$, $a_{10}^{e}$, $a_{22}^{e}$ and $a_{25}^{e}$ have degree $2$ in $G$. 
	Let $S$ be the set of all maximal subpaths of $C$ which only contain vertices from $\{a_{21}^{e},\dots,a_{26}^{e},a_6^{\tilde{e}},\dots,a_{11}^{\tilde{e}}: e\in E_{\vec{\mathcal{E}}}^{-}(v),\tilde{e}\in E_{\vec{\mathcal{E}}}^{+}(v)\}$. Since there  are no edges of the form $\{a_i^{e},a_j^{\tilde{e}}\}$ for $i,j\in \{6,\dots,11,21,\dots,26\}$, $e\not= \tilde{e}\in E(\vec{\mathcal{E}})$, every subpath in $S$ is either of length $2$ or length $5$. 
	For every path $P=(p_1,\dots,p_\ell)\in S$, we define the vertices $u_P,w_P$ to be the neighbours of $P$ on $C$, i.e.\ $(u_P,p_1,\dots,p_\ell,w_P)$ is a subpath of $C$. 
	Since $G_\mathcal{E}[N_{G_\mathcal{E}}(S_v)]\cong G[N_G(S_v)]$ and every path $P\in S$ is maximal, we know that $u_P,w_P\in \{a_{18}^{e},a_{20}^{e},a_{27}^{e},a_{29}^{e},a_3^{\tilde{e}},a_5^{\tilde{e}},a_{12}^{\tilde{e}},a_{14}^{\tilde{e}}: e\in E_{\vec{\mathcal{E}}}^{-}(v),\tilde{e}\in E_{\vec{\mathcal{E}}}^{+}(v) \}\cup\{b_3^v,b_4^v\}$. 
	
	Properties (\ref{order1}),(\ref{order3}) imply that for every edge $e\in E_{\vec{\mathcal{E}}}^{-}(v)$, only one of the two vertices $a_{18}^{e},a_{20}^{e}$  and only one of the two vertices $a_{27}^{e},a_{29}^{e}$ can be in the set $\{u_P,w_P: P\in S\}$. Therefore the contribution of the set $\{a_{18}^{e},a_{20}^{e},a_{27}^{e},a_{29}^{e}: e\in E_{\vec{\mathcal{E}}}^{-}(v)\}$ to the size of $\{u_P,w_P: P\in S\}$ is at most $2|E^-_{\vec{{\cal E}}}(v)|$.
	Similarly, (\ref{order2}),(\ref{order3}) imply that for every edge  $\tilde{e}\in E_{\vec{\mathcal{E}}}^{+}(v)$ only one of the two vertices $a_{3}^{\tilde{e}},a_5^{\tilde{e}}$  and only one of the two vertices $a_{12}^{\tilde{e}},a_{14}^{\tilde{e}}$ can be in the set $\{u_P,w_P: P\in S\}$. 
	In addition, there are two not necessarily distinct edges $e,\tilde{e}\in E_{\vec{\mathcal{E}}}^{+}(v)$ such that $(a_1^{e},\dots,a_5^{e},b_1^v)$ and $(b_6^v,a_{12}^{\tilde{e}},\dots,a_{20}^{\tilde{e}})$ are subpaths of $C$ since $b_1^v$'s and $b_6^v$'s only other neighbour in $G_{\mathcal{E}}$ is $b_2^v$ or $b_5^v$ respectively. Hence the vertices $a_{3}^{e},a_5^{e},a_{12}^{\tilde{e}},a_{14}^{\tilde{e}}$ cannot be in $\{u_P,w_P: P\in S\}$ for these two specific edges $e,\tilde{e}\in E^-_{\vec{{\cal E}}}(v)$. Therefore, the contribution of the set $\{a_3^{\tilde{e}},a_5^{\tilde{e}},a_{12}^{\tilde{e}},a_{14}^{\tilde{e}}: \tilde{e}\in E_{\vec{\mathcal{E}}}^{+}(v) \}$ to the size of $|\{u_P,w_P: P\in S\}$ is at most $2|E^+_{\vec{{\cal E}}}(v)|-2$.
	Hence $|\{u_P,w_P: P\in S\}|\leq 2|E_{\vec{\mathcal{E}}}^{+}(v)|-2+2|E_{\vec{\mathcal{E}}}^{-}(v)|+2=4d$ where adding $2$ accounts for the size of the set $\{b_3^v,b_4^v\}$. 
	In addition, note that (\ref{order1}),(\ref{order2}), (\ref{order3}) and $\deg_G(b_2^v)=2$ and $\deg_G(b_5^v)=2$ imply that every maximal subpath of $C$ only containing vertices in $S_v\setminus \{a_{21}^{e},\dots,a_{26}^{e},a_6^{\tilde{e}},\dots,a_{11}^{\tilde{e}}: e\in E_{\vec{\mathcal{E}}}^{-}(v),\tilde{e}\in E_{\vec{\mathcal{E}}}^{+}(v)\}$ has length at least $2$. Hence on $C$ in between any two subpaths in $S$ there is a subpath of at length at least $2$, which implies that no two distinct subpaths in $S$ can share a neighbour on $C$. Therefore we have  $|\{u_P,w_P: P\in S\}|=2|S|$ and hence $|S|\leq 2d$. 
	
	On the other hand, if a path in $S$ contains only $3$ vertices (has length $2$) then $|S|> 2d$, because every path in $S$ contains either $3$ or $6$ vertices (as argued at the beginning of the proof), no vertex can appear on more than one path of $S$ and $|\{v: P\in S, v \text{ appears on }P\}|=|\{a_{21}^{e},\dots,a_{26}^{e},a_6^{\tilde{e}},\dots,a_{11}^{\tilde{e}}: e\in E_{\vec{\mathcal{E}}}^{-}(v),\tilde{e}\in E_{\vec{\mathcal{E}}}^{+}(v)\}|=12d$. This yields a contradiction and hence (\ref{order4}) and (\ref{order5}) are true.  
\end{proof}

Let $G$ be a graph with $a_i^e,\dots,a_{j}^e\in V(G)$ for some edge $e\in E(\vec{\mathcal{E}})$ and $1\leq i \leq j\leq 31$. Assume $C$ is a cycle in $G$ which contains $a_i^e,\dots,a_j^e$. We say that $C$ \emph{traverses} the vertices $a_i^e,\dots,a_j^e$  \emph{in order} if $(a_i^e,\dots,a_j^e)$ is a subpath of $C$ and we say that $C$ \emph{traverses $a_i^e,\dots,a_j^e$ out of order} otherwise. 
Note that for certain $1\leq i\leq j\leq 31$ and $e\in E(\vec{\mathcal{E}})$   there is only one way in which a cycle $C$ can traverse  $a_i^e,\dots,a_j^e$ out of order (as specified in Lemma~\ref{lem:orderOfBeginningMiddleAndEndPieceOfP(...)}).

The next lemma shows that for every vertex $v\in V(\vec{\mathcal{E}})$ and every Hamiltonian cycle $C$ in $G_\mathcal{E}$ the number of edges $e\in E_{\vec{\mathcal{E}}}^{-}(v)$ for which $C$ traverses $a_{12}^{e},\dots,a_{20}^{e}$ out of order is exactly one larger than the number of edges $\tilde{e}\in E_{\vec{\mathcal{E}}}^{+}(v)$ for which $C$ traverses $a_{12}^{\tilde{e}},\dots,a_{20}^{\tilde{e}}$ out of order. This still holds for every graph $G$ which contains a certain induced subgraph of $G_\mathcal{E}$.  
\begin{lemma}\label{lem:inAndOutSetEqualCardinality}
	Let $\mathcal{E}$ be any $d$-regular graph and  $G_\mathcal{E}$ as defined in Definition~\ref{def:G_E}.
	Let $S_v:=\{a_1^e,\dots, a_{31}^e: e\in E(\vec{\mathcal{E}}),e\text{ is incident to }v\}\cup \{b_1^v,\dots,b_6^v\}$ for some $v\in V(\vec{\mathcal{E}})$. 
	Let $G$ be a graph with $S_v\subseteq V(G)$. Assume  $G_\mathcal{E}[N_{G_\mathcal{E}}(S_v)]\cong G[N_G(S_v)]$ and $f: S_v\rightarrow S_v$ defined by $f(v)=v$ for $v\in S_v$ is an isomorphism from $G_\mathcal{E}[S_v]$ to $G[S_v]$.
	Then for every Hamiltonian cycle $C$ in $G$  the cardinalities of the two sets
	\begin{align}
	&\inSet:=\Big\{e\in E_{\vec{\mathcal{E}}}^{-}(v): (a_{12}^{e},a_{17}^{e})\text{ is a subpath of }C\Big\}\text{ and }\label{eq:inSet}\\
	&\outSet:=\Big\{e\in E_{\vec{\mathcal{E}}}^{+}(v): (a_{12}^{e},a_{17}^{e})\text{ or }(a_{12}^{e},b_6^v)\text{ is a subpath of }C\Big\}\label{eq:outSet}
	\end{align} 
	are equal. 
\end{lemma}
\begin{proof}
	First note that the condition $G_\mathcal{E}[N_{G_\mathcal{E}}(S_v)]\cong G[N_G(S_v)]$ implies that no vertex in the set $\{a_{15}^{e},\dots,a_{30}^{e},a_2^{\tilde{e}},\dots,a_{17}^{\tilde{e}}: e\in E_{\vec{\mathcal{E}}}^{-}(v),\tilde{e}\in E_{\vec{\mathcal{E}}}^{+}(v)\}\cup \{b_1^v,\dots,b_6^v\}$ has neighbours  in $G\setminus S_v$. This will implicitly be used in the following argument whenever we exhaustively consider neighbours of vertices in $G$ as successors on $C$.
	
	Let us first define a map $F_{v,C}:\inSet\rightarrow \outSet$, given by $F_{v,C}(e):=\tilde{e}$, where $\tilde{e}\in \outSet$ is the edge such that $(a_{18}^e,a_8^{\tilde{e}})$ is a subpath of $C$. We first have to argue that $F_{v,C}$ is well defined. 
	
	By Lemma~\ref{lem:orderOfBeginningMiddleAndEndPieceOfP(...)}~(\ref{order3}), $e\in \inSet$ implies   that $(a_{14}^e,a_{13}^e,a_{12}^e,a_{17}^e,a_{16}^e,a_{15}^e,a_{20}^e,a_{19}^e,a_{18}^e)$ is a subpath of $C$. Since the two neighbours $a_{17}^e$ and $a_{19}^e$ of $a_{18}^e$ are already part of this subpath this implies that $(a_{18}^e,a_8^{\tilde{e}})$ has to be a subpath of $C$ for some edge $\tilde{e}\in E_{\vec{\mathcal{E}}}^{+}(v)$. This implies that $(a_6^{\tilde{e}},\dots,a_{11}^{\tilde{e}})$ cannot be a subpath of $C$ and hence, by Lemma~\ref{lem:orderOfBeginningMiddleAndEndPieceOfP(...)}~(\ref{order4}),
	$(a_8^{\tilde{e}},a_7^{\tilde{e}},a_6^{\tilde{e}},a_{11}^{\tilde{e}},a_{10}^{\tilde{e}},a_{9}^{\tilde{e}})$ has to be a subpath of $C$. This further implies that $(a_{11}^{\tilde{e}},a_{12}^{\tilde{e}})$ cannot be a subpath of $C$. Then if $(a_{12}^{\tilde{e}},\dots, a_{20}^{\tilde{e}})$ is a subpath of $C$ then $(a_{12}^{\tilde{e}},b_6^v)$ has to be a subpath of $C$ by excluding all possible other neighbours of $a_{12}^{\tilde{e}}$. On the other hand if $(a_{12}^{\tilde{e}},\dots, a_{20}^{\tilde{e}})$ is not a subpath of $C$ then, by Lemma~\ref{lem:orderOfBeginningMiddleAndEndPieceOfP(...)}~(\ref{order3}),  $(a_{14}^{\tilde{e}},a_{13}^{\tilde{e}},a_{12}^{\tilde{e}},a_{17}^{\tilde{e}},a_{16}^{\tilde{e}},a_{15}^{\tilde{e}},a_{20}^{\tilde{e}},a_{19}^{\tilde{e}},a_{18}^{\tilde{e}})$ is a subpath of $C$ and hence $(a_{12}^{\tilde{e}},a_{17}^{\tilde{e}})$ is a subpath of $C$. Therefore $\tilde{e}\in \outSet$. This shows that $F_{v,C}$ is well defined.
	
	Furthermore $F_{v,C}$ is injective since if $(a_{18}^e,a_8^{\tilde{e}})$ and $(a_{18}^e,a_8^{e'})$ are subpaths of $C$ then $\tilde{e}=e'$ because $(a_{19}^e,a_{18}^e)$ is also a subpath of $C$. $F_{v,C}$ is surjective as for $\tilde{e}\in \outSet$ both  $(a_{12}^{\tilde{e}},a_{17}^{\tilde{e}})$ or $(a_{12}^{\tilde{e}},b_6^v)$ being a subpath of $C$ together with Lemma~\ref{lem:orderOfBeginningMiddleAndEndPieceOfP(...)}~(\ref{order3}) implies that $(a_{12}^{\tilde{e}},a_{11}^{\tilde{e}})$ cannot be a subpath of $C$. This further implies that  $(a_8^{\tilde{e}},a_7^{\tilde{e}},a_6^{\tilde{e}},a_{11}^{\tilde{e}},a_{10}^{\tilde{e}},a_{9}^{\tilde{e}})$ is a subpath of $C$ by Lemma~\ref{lem:orderOfBeginningMiddleAndEndPieceOfP(...)}~(\ref{order4}) and hence there is an edge $e\in E_{\vec{\mathcal{E}}}^{-}(v)$ such that $(a_{18}^e,a_{8}^{\tilde{e}})$ is a subpath of $C$. Then with the same argument as before $(a_{12}^e,a_{17}^e)$ is a subpath of $C$ and hence $e\in \inSet$ and $F_{v,C}(e)=\tilde{e}$. Therefore $F_{v,C}$ is bijective which implies the statement of the lemma.
\end{proof}

As a direct consequence from Lemma~\ref{lem:inAndOutSetEqualCardinality} we get that $G_\mathcal{E}$ cannot be Hamilonian for any base graph $\mathcal{E}$. 
To see that this is true, suppose towards a contradiction that there is a Hamiltonian cycle $C$ in $G_\mathcal{E}$. Then by Lemma~\ref{lem:inAndOutSetEqualCardinality} the  equation $\sum_{v\in V(\vec{\mathcal{E}})}|\inSet|=\sum_{v\in V(\vec{\mathcal{E}})}|\outSet|$ must hold.
Now every edge in $\bigcup_{v\in V(\vec{\mathcal{E}})}\inSet$ is also contained in $\bigcup_{v\in V(\vec{\mathcal{E}})}\outSet$. However, since for every $v\in V({\cal E})$ the vertex $b_6^v\in G_\mathcal{E}$ has only one neighbour not of the form $a_{12}^{(v,w)}$, there must be edges  $(v,w)$ in $\vec{{\cal E}}$ for which $(a_{12}^{(v,w)}, b_6^v)$ is a subpath of $C$. Hence	
$\bigcup_{v\in V(\vec{\mathcal{E}})}\outSet$ must contain some edges (all the edges $(v,w)$ for which $(a_{12}^{(v,w)},b_6^v)$ is a subpath of $C$) that are not contained in $\bigcup_{v\in V(\vec{\mathcal{E}})}\inSet$, so the equation cannot hold, a contradiction. Hence $G_\mathcal{E}$ cannot be Hamiltonian. This argument works similarly if a small number of edges in $G_\mathcal{E}$ have been altered and the equality from Lemma~\ref{lem:inAndOutSetEqualCardinality} still has to hold for many vertices as we will see in the following proof. 

\begin{proof}[Proof of Theorem~\ref{thm:farFromHam}]
	Let $\epsilon:={1}/(8(d+3)^2(6+31d))$. Assume $\mathcal{E}$ is $d$-regular and $n:=|V(\mathcal{E})|$. Let $n':=V(G_\mathcal{E})=(6+31d)n$ and $d':=d+3$ be the degree of $G_\mathcal{E}$.
	
	Towards a contradiction let us assume that $G_\mathcal{E}$ is not $\epsilon$-far to being Hamiltonian and let $E$ be a set of edges in $G_{\cal E}$ such that $|E|\leq \epsilon d'n'$ and the graph $G:=(V(G_\mathcal{E}),E(G_\mathcal{E})\triangle E)$ is Hamiltonian. Let $B\subseteq V(\vec{\mathcal{E}})$ be the set of vertices defined by 
	\begin{align*}
	B:=&\{v\in V(\vec{\mathcal{E}}): \text{there is }e\in E,i\in \{1,\dots,31\},\tilde{e}\in E_G^{-}(v)\cup E_G^{+}(v)\text{ such that }a_i^{\tilde{e}}\in e  \}\\ \cup &\{v\in V(\vec{\mathcal{E}}): \text{there is }e\in E,i\in \{1,\dots,6\} \text{ such that }b_i^v\in e   \}.
	\end{align*}
	Note that $|B|\leq 4\cdot \epsilon d'n'$, because every edge $e\in E$ contributes at most $4$ vertices to $B$, and hence $|V(\vec{\mathcal{E}})\setminus B|\geq n-4\epsilon d'n'>{n}/{2}$.
	
	Let $C$ be a Hamiltonian cycle in $G$. Then for every vertex $v\in V(\vec{\mathcal{E}})\setminus B$ we have that $S_v\subseteq V(G)$,  $G_\mathcal{E}[N_{G_\mathcal{E}}(S_v)]\cong G[N_G(S_v)]$ and $f: S_v\rightarrow S_v$ defined by $f(v)=v$ for $v\in S_v$ is an isomorphism from $G_\mathcal{E}[S_v]$ to $G[S_v]$ where $S_v:=\{a_1^e,\dots,a_{31}^e: e\in E(\vec{\mathcal{E}}),e\text{ is incident to }v\}\cup \{b_1^v,\dots,b_6^v\}$. Since $C$ is Hamiltonian $C$ contains all vertices in $S_v$ for every $v\in V(\vec{\mathcal{E}})\setminus B$ (amongst others). Hence by Lemma~\ref{lem:inAndOutSetEqualCardinality} we have $|\inSet|=|\outSet|$ for every $v\in V(\vec{\mathcal{E}})\setminus B$ where $\inSet$ and $\outSet$ are as defined in Equation~\ref{eq:inSet} and  Equation~\ref{eq:outSet}. 
	Therefore 
	\begin{align}\label{eq:sumInAndOutSets}
	\sum_{v\in V(\vec{\mathcal{E}})\setminus B}|\inSet|=\sum_{v\in V(\vec{\mathcal{E}})\setminus B}|\outSet|.
	\end{align}
	Since for every $v\in V(\vec{\mathcal{E}})\setminus B$ the vertex $b_5^v$ has degree $2$ we know that $(b_4^v,b_5^v,b_6^v)$ is a subpath of $C$. Since every neighbour of $b_6^{v}$ in $G$ apart from $b_5^v$ is of the form $a_{12}^{e}$ for some $e\in E_G^{+}(v)$, we know that for precisely one edge $e\in E_G^{+}(v)$ the path $(b_6^v,a_{12}^{e})$ is a subpath of $C$. Hence  
	\begin{align}\label{eq:upperBoundOnDiagonalEdges}
	\sum_{v\in V(\vec{\mathcal{E}})\setminus B}\Big|\Big\{e\in E_G^{+}(v): (a_{12}^{e},b_6^v)\text{ is a subpath of }C\Big\}\Big|= |V(\vec{\mathcal{E}})\setminus B|>\frac{n}{2}.
	\end{align}
	Since every edge $(u,v)\in E(\vec{\mathcal{E}})$ such that $u,v\in V(\vec{\mathcal{E}})\setminus B$ contributes $1$ to both sides of Equation~\ref{eq:sumInAndOutSets}, Equation~\ref{eq:sumInAndOutSets} and Equation~\ref{eq:upperBoundOnDiagonalEdges} imply that 
	\begin{align*}
	\sum_{v\in V(\vec{\mathcal{E}})\setminus B}\Big|\Big\{(u,v)\in E(\vec{\mathcal{E}}): u\in B,(a_{12}^{(u,v)},a_{17}^{(u,v)})\text{ is a subpath of }C\Big\}\Big|>\frac{n}{2}.
	\end{align*} 
	But this is a contradiction as the number of edges $(u,v)\in E(\vec{\mathcal{E}})$ for which $u\in B$ is bounded from above by $d'|B|\leq {n}/{2}$.
\end{proof}

\section{Ensuring local Hamiltonicity}\label{sec:localHAM}
In this section we prove the following theorem.
\begin{theorem}\label{thm:locallyHam}
	For every $d\in \mathbb{N}_{\geq 2}$ and every $d$-regular graph $\mathcal{E}$  with expansion ratio $h(\mathcal{E})\geq 1$ the graph $G_\mathcal{E}$  constructed in Definition~\ref{def:G_E} is $\delta$-locally Hamiltonian for some constant $\delta=\delta(d)\in (0,1]$.
\end{theorem}
Our proof strategy for Theorem~\ref{thm:locallyHam} is  to add edges to $G_\mathcal{E}$ which are incident to at most one vertex in $N_{G_\mathcal{E}}(S)$ to obtain a graph $H$ which is Hamiltonian, for any given $S\subseteq V(G_\mathcal{E})$ of size at most $\delta |V(G)|$. We prove the Hamiltonicity of $H$ by dividing the vertex set of $H$ into pairwise disjoint small sets. For each of these sets  we obtain a set of vertex disjoint paths which cover the entire small set and start and end in prescribed vertices. To conclude the proof of the Hamiltonicity of $H$ we find a Hamiltonian cycle by patching together these paths. The next lemma will be used to show the existence of such paths for all those subsets of vertices of $H$ which contain a vertex from $S$.
\begin{lemma}\label{lem:existenceOfEdgeDisjoinPaths}
	Let $\mathcal{E}$ be any $d$-regular graph and  $G_\mathcal{E}$ as defined in Definition~\ref{def:G_E}.
	Let  $v\in V(\vec{\mathcal{E}})$ and $S_v:=\{a_{18}^{e},\dots,a_{31}^{e},a_1^{\tilde{e}},\dots,a_{17}^{\tilde{e}}: e\in E_{\vec{\mathcal{E}}}^{-}(v),\tilde{e}\in E_{\vec{\mathcal{E}}}^{+}(v)\}\cup \{b_1^v,\dots,b_6^v\}$. 
	Let $G$ be a graph such that $G_\mathcal{E}[S_v]$ is a subgraph of $G$.
	Then for any two sets $T_v^{\operatorname{in}}\subseteq E_{\vec{\mathcal{E}}}^{-}(v)$ and $T_v^{\operatorname{out}}\subseteq E_{\vec{\mathcal{E}}}^{+}(v)$ with $|T_v^{\operatorname{in}}|-1=|T_v^{\operatorname{out}}|$ there is a set of $2d$ pairwise vertex disjoint simple paths $\{P_{e}^{\operatorname{in}},P_{\tilde{e}}^{\operatorname{out}}: e\in E_{\vec{\mathcal{E}}}^{-}(v),\tilde{e}\in E_{\vec{\mathcal{E}}}^{+}(v)\}$ in $G$ with the following properties.
	\begin{itemize}
		\item If $e\in T_v^{\operatorname{in}}$ then $P_e^{\operatorname{in}}$ is a path from $a_{20}^e$ to $a_{31}^e$.
		\item If $e\in E_{\vec{\mathcal{E}}}^{-}(v)\setminus T_v^{\operatorname{in}}$ then $P_e^{\operatorname{in}}$ is a path from $a_{18}^e$ to $a_{31}^e$.
		\item If $e\in T_v^{\operatorname{out}}$ then $P_e^{\operatorname{out}}$ is a path from $a_{1}^e$ to $a_{15}^e$.
		\item If $e\in E_{\vec{\mathcal{E}}}^{+}(v)\setminus T_v^{\operatorname{out}}$ then $P_e^{\operatorname{out}}$ is a path from $a_{1}^e$ to $a_{17}^e$.
		\item The set  $\{x\in V(G): x\text{ is contained in }P_e^{\operatorname{in}}\text{ or }P_e^{\operatorname{out}}\text{ for some }e \}$ is equal to $S_v$.
	\end{itemize}  
\end{lemma}
\begin{proof}
	First we pick a vertex $n(v)\in V(\vec{\mathcal{E}})$ such that $(v,n(v))\notin T_v^{\operatorname{out}}$. This is possible because $v$ has the same number of incoming and outgoing edges and $|T_v^{\operatorname{in}}|-1=|T_v^{\operatorname{out}}|$. Then $|T_v^{\operatorname{in}}|=|T_v^{\operatorname{out}}\cup \{(v,n(v))\}|$, and hence we can find a bijection $g:T_v^{\operatorname{in}}\rightarrow T_v^{\operatorname{out}}\cup \{(v,n(v))\}$. Then we can define the paths as follows. For $e\in T_v^{\operatorname{in}}$ we let
	\begin{displaymath}
	P_{e}^{\operatorname{in}}:= (a_{20}^{e},a_{19}^{e},a_{18}^{e}, a_{8}^{g(e)},a_{7}^{g(e)},a_{6}^{g(e)},a_{11}^{g(e)},a_{10}^{g(e)},a_{9}^{g(e)},a_{29}^{e},a_{28}^{e},a_{27}^{e},a_{30}^{e},a_{31}^{e}),
	\end{displaymath}
	\begin{displaymath}
	P_{g(e)}^{\operatorname{out}}:= (a_{1}^{g(e)},\dots,a_{5}^{g(e)},b_1^v,b_2^v,b_3^v, a_{23}^{e},a_{22}^{e},a_{21}^{e},a_{26}^{e},a_{25}^{e},a_{24}^{e},b_4^v,b_5^v,b_6^v,a_{12}^{g(e)},\dots,a_{17}^{g(e)}) 
	\end{displaymath}
	if $g(e)=(v,n(v))$ and 
	\begin{align*}
	P_{g(e)}^{\operatorname{out}}:= (a_{1}^{g(e)},a_{2}^{g(e)},a_{5}^{g(e)},a_{4}^{g(e)},a_{3}^{g(e)}, a_{23}^{e},a_{22}^{e},a_{21}^{e},a_{26}^{e},a_{25}^{e},a_{24}^{e},a_{14}^{g(e)},\hspace*{2pt} & a_{13}^{g(e)},a_{12}^{g(e)},a_{17}^{g(e)},
	a_{16}^{g(e)},a_{15}^{g(e)}) 
	\end{align*}
	if $g(e)\not=(v,n(v))$.
	
	Furthermore for  $e\in E_{\vec{\mathcal{E}}}^{-}(v)\setminus T_v^{\operatorname{in}}$ we let $P_e^{\operatorname{in}}:=(a_{18}^e,\dots,a_{31}^e)$ and for $e\in E_{\vec{\mathcal{E}}}^{+}(v)\setminus T_v^{\operatorname{out}}$ we let $P_e^{\operatorname{out}}:=(a_{1}^e,\dots,a_{17}^e)$. These paths clearly satisfy all conditions.
	
\end{proof}

\begin{proof}[Proof of Theorem~\ref{thm:locallyHam}]
	Let $\delta:={1}/(2\cdot(6+31d))$ and let $S\subseteq V(G_\mathcal{E})$ be any set of vertices with $|S|\leq \delta\cdot |V(G_\mathcal{E})|$. We will find a Hamiltonian graph $H$ by modifying $G_\mathcal{E}$ in such a way that $G_\mathcal{E}[N_{G_\mathcal{E}}(S)]$ is not affected by any modifications. In the following we exclude the trivial case $S=\emptyset$.
	
	Let $S_v:=\{a_{18}^{e},\dots,a_{31}^{e},a_1^{\tilde{e}},\dots,a_{17}^{\tilde{e}}: e\in E_{\vec{\mathcal{E}}}^{-}(v),\tilde{e}\in E_{\vec{\mathcal{E}}}^{+}(v)\}\cup \{b_1^v,\dots,b_6^v\}$ for every  $v\in V(\vec{\mathcal{E}})$. 
	Let $S':=\{v\in V(\vec{\mathcal{E}}): S_v\cap S\not= \emptyset \}$.  By Remark~\ref{rem:degreeAndNumberOfVerticesOFG_E} $|V(G_\mathcal{E})|=(6+31d)\cdot |V(\mathcal{E})|$. Since the sets $S_v$ are pairwise disjoint this implies that $|S'|\leq |S|\leq \delta\cdot |V(G_\mathcal{E})|={1}/{2}\cdot |V(\mathcal{E})|$.  Let $S'=\{s_1',\dots,s_m'\}$ where $m:=|S'|$.
	\begin{claim}\label{claim:edgeDisjointPath} 
		There are pairwise edge disjoint paths $Q_1,\dots,Q_m$ in $\mathcal{E}$ such that  $Q_i$ is of the form $Q_i=(q_i^1,\dots,q_i^{\ell_i})$ for some $\ell_i\in \mathbb{N}$ and $q_i^{\ell_i}=s_i'$, $q_i^j\in S'$ for all $j>1$ and $q_i^{1}\in V(\mathcal{E})\setminus S'$.
	\end{claim}
	\begin{proof}
		By induction on the size of $S'$. If $|S'|=1$ then this is trivially true. If $|S'|=n$ then $h(\mathcal{E})\geq 1$ implies that there must be a vertex $v$ with at least as many neighbours in $V(\mathcal{E})\setminus S'$ as neighbours in $S'$. Then $S\setminus \{v\}$ has $n-1$ vertices. Hence by induction hypothesis there is such a set of paths for $S'\setminus\{v\}$. But then we can extend every path which starts in $v$ by a different  edge so it starts in $V(\mathcal{E})\setminus S$. 
	\end{proof}
	Let $Q_1,\dots, Q_m$ be as in Claim~\ref{claim:edgeDisjointPath}. Further, for every vertex $v\in V(\mathcal{E})\setminus S'$ we pick one vertex $u\in V(\mathcal{E})$ with $(v,u)\in E(\vec{\mathcal{E}})$ and define $n(v):=u$.
	Now let $E$ be the set
	\begin{align*}
	\bigg\{\Big\{b_3^v,a_{4}^{(v,n(v))}\Big\},\Big\{b_4^v,a_{13}^{(v,n(v))}\Big\}: v\in V(\mathcal{E})\setminus S'\bigg\}\cup \bigg\{\Big\{a_{14}^{(q_i^1,q_i^2)},a_{17}^{(q_i^1,q_i^2)}\Big\}: 1\leq i\leq m\bigg\}.
	\end{align*}
	We now define the graph $H$ by setting $V(H):=V(G_\mathcal{E})$ and $E(H):=E(G_\mathcal{E})\cup E$. Note that $H$ has degree $d+3$, as we only added at most one edge to vertices of degree at most $d+1$. Further note that by definition of $S'$ we have that $S\subseteq \bigcup_{v\in S'}S_v$. Since every edge in $E$ is incident to at most one vertex in $N_G(\bigcup_{v\in S'}S_v)$ it follows that if $H$ is Hamiltonian then it fulfils the conditions from Definition~\ref{def:locallyHam}. 
	Therefore, if we prove that $H$ has a Hamiltonian cycle then $G_\mathcal{E}$ must be locally Hamiltonian. 	
	\begin{claim}\label{claim:pairwiseDisjointPaths}
		There is a set of $2d$ pairwise vertex disjoint simple paths $\{P_{e}^{\operatorname{in}},P_{\tilde{e}}^{\operatorname{out}}: e\in E_{\vec{\mathcal{E}}}^{-}(v),\tilde{e}\in E_{\vec{\mathcal{E}}}^{+}(v)\}$ for every $v\in V(\vec{\mathcal{E}})\setminus S'$ with the following properties.
		\begin{itemize}
			\item If $e\in E_{\vec{\mathcal{E}}}^{-}(v)$ then $P_{e}^{\operatorname{in}}$ is a path from $a_{18}^{e}$ to $a_{31}^{e}$.
			\item If $e=(q_i^{1},q_i^{2})$ for some $i\in [m]$ then $P_{e}^{\operatorname{out}}$ is a path from $a_{1}^{e}$ to $a_{15}^{e}$.
			\item If $e\in E_{\vec{\mathcal{E}}}^{+}(v)\setminus \{(q_i^{1},q_i^{2}): i\in [m]\}$ then $P_{e}^{\operatorname{out}}$ is a path from $a_{1}^{e}$ to $a_{17}^{e}$.
			\item The set  $\{x\in V(G): x\text{ is contained in }P_e^{\operatorname{in}}\text{ or }P_e^{\operatorname{out}}\text{ for some }e \}$ is equal to $S_v$.
		\end{itemize} 
	\end{claim}
	\begin{proof}
		These conditions can be achieved by letting $P_{e}^{\operatorname{in}} :=(a_{18}^{e},\dots,a_{31}^e)$ for $e\in E_{\vec{\mathcal{E}}}^{-}(v)$. Additionally, for every edge $e=(q_i^1,q_i^2)$ we let $P_{e}^{\operatorname{out}} :=(a_{1}^{e},\dots,a_{14}^e,a_{17}^e,a_{16}^e,a_{15}^e)$ if $q_i^2\not= n(q_i^1)$  and  $P_{e}^{\operatorname{out}}=(a_1^e,\dots,a_4^e,b_3^v,b_2^v,b_1^v,a_5^e,\dots,a_{12}^e,b_6^v,b_5^v,b_4^v,a_{13}^e,a_{14}^e,a_{17}^e,a_{16}^e,a_{15}^e)$ otherwise. Finally for every $e\in E_{\vec{\mathcal{E}}}^{+}(v)\setminus \{(q_i^{1},q_i^{2}): i\in [m]\}$ we set $P_e^{\operatorname{out}}:=(a_1^e,\dots,a_{17}^e)$ for $e=(v,w)$, $w\not=n(v)$ and $P_{e}^{\operatorname{out}}:=(a_1^e,\dots,a_4^e,b_3^v,b_2^v,b_1^v,a_5^e,\dots,a_{12}^e,b_6^v,b_5^v,b_4^v,a_{13}^e,\dots,a_{17}^e)$ for $e=(v,n(v))$.
	\end{proof}
	
	For $v\in S'$ we define the sets $T_v^{\operatorname{in}}:=\{(q_i^{j-1},q_i^{j}): i\in [m],j\in \{2,\dots, \ell_i\},q_i^{j}=v \}$ and $T_v^{\operatorname{out}}:=\{(q_i^j,q_i^{j+1}): i\in [m],j\in \{2,\dots,\ell_i-1\},q_i^{j}=v \}$. Since for every $v\in S'$ there is exactly one path out of $Q_1,\dots,Q_m$ that ends in $v$, we get that $|T_v^{\operatorname{in}}|-1=|T_v^{\operatorname{out}}|$ and hence the preconditions for Lemma~\ref{lem:existenceOfEdgeDisjoinPaths} are met. Therefore we obtain a set of paths $\{P_{e}^{\operatorname{in}},P_{\tilde{e}}^{\operatorname{out}}: e\in E_{\vec{\mathcal{E}}}^{-}(v),\tilde{e}\in E_{\vec{\mathcal{E}}}^{+}(v)\}$ for every $v\in S'$ as in Lemma~\ref{lem:existenceOfEdgeDisjoinPaths}.
	
	Since $S_v\cap S_w=\emptyset$ for every pair $v,w\in V(\vec{\mathcal{E}})$ with $v\neq w$, we now have a set of pairwise vertex  disjoint simple paths $\{P_{e}^{\operatorname{in}},P_{e}^{\operatorname{out}}: e\in E(\vec{\mathcal{E}})\}$ such that every vertex of $H$ is contained in one of the paths. 
	For every edge $e\in E(\vec{\mathcal{E}})$ we now concatenate $P_e^{\operatorname{out}}$ with $P_e^{\operatorname{in}}$ to a path $P_e$. This is possible as for every edge $e\in E(\vec{\mathcal{E}})$ the end vertex of $P_e^{\operatorname{out}}$ and the start vertex of $P_e^{\operatorname{in}}$ are adjacent.
	Let us briefly explain why they are adjacent.  This is clearly true for every $e$ which does not appear on any path  $Q_1,\dots,Q_m$, because in this case $P_e^{\operatorname{out}}$ ends in $a_{17}^e$ and $P_e^{\operatorname{in}}$ starts in $a_{18}^e$. In the case that $e=(q_i^{j},q_i^{j+1})$ for some $i\in [m]$ and $j\in \{2,\dots, \ell_i-1\}$ we have that $P_e^{\operatorname{out}}$ ends in $a_{15}^e$ and $P_e^{\operatorname{in}}$ starts in $a_{20}^e$ which are adjacent in $H$. This leaves the case that $e=(q_i^{1},q_i^{2})$ for some $i\in [m]$. Since $q_i^1\in V(\mathcal{E})\setminus S'$  we get that $P_e^{\operatorname{out}}$ ends in $a_{15}^e$ and since $e\in T_{q_i^2}^{\operatorname{in}}$ we get that $P_e^{\operatorname{in}}$ starts in $a_{20}^e$.
	
	Finally we concatenate all paths $P_e$ in the order given by the ordering $f:E(\vec{\mathcal{E}})\rightarrow [|E(\vec{\mathcal{E}})|]$ used in the construction of $G_\mathcal{E}$. This gives us a cycle which contains every vertex in $H$ precisely once. Hence $H$ is Hamiltonian.
	
\end{proof} 

\begin{theorem}\label{thm:existenceOfLocallyButFarFromHamiltonianGraphs}
	There are $d\in \mathbb{N}$ and constants  $\delta:=\delta(d),\epsilon:=\epsilon(d) \in (0,1)$ and a sequence of  $d$-bounded degree graphs $(G_N)_{N\in \mathbb{N}}$ of increasing order such that $G_N$ is $\delta$-locally Hamiltonian and $\epsilon$-far from being Hamiltonian for every $N\in \mathbb{N}$. 
\end{theorem}
\begin{proof}
	Let $D\in \mathbb{N}$ and $(\mathcal{E}_N)_{N\in \mathbb{N}}$ a sequence of $D$-bounded degree  graphs with expansion ratio at least $1$ of increasing order. For explicit constructions of such expanders, see for example~\cite{MargulisConstruction} or~\cite{ZigZagProductIntroduction}. Then for every $N\in \mathbb{N}$ we set $G_N:=G_{\mathcal{E}_N}$ be the graph constructed in Definition~\ref{def:G_E}. By Theorem~\ref{thm:farFromHam} and Theorem~\ref{thm:locallyHam} there is a degree bound $d$ and  constants $\delta,\epsilon \in (0,1)$, whose size only depends on $D$, such that $G_N$ has degree bounded by $d$ and $G_N$ is $\delta$-locally Hamiltonian and $\epsilon$-far from being Hamiltonian. 
\end{proof}
\section{Application to property testing}\label{sec:applicationsToPT}
In this section we introduce the bounded-degree model of property testing as introduced in \cite{GoldreichRon2002} and then use our main result from Section~\ref{sec:localHAM} to prove the known sublinear lower bound for the complexity of property testing Hamiltonicity \cite{DBLP:journals/ieicet/YoshidaI10,Goldreich20} for one-sided error testers.

Let $d\in \mathbb N$ and let $\mathcal{C}_d$ be the class of graphs of bounded degree $d$. From now on, all  graphs have $d$-bounded degree. A property $\mathcal{P}$ on $\mathcal{C}_d$ is any subset of $\mathcal{C}_d$ which is closed under isomorphism.  
An algorithm that processes a graph $G$
does not obtain an encoding of $G$ as a bit string in the usual way.  Instead,
it has direct access to $G$ using an \emph{oracle} which answers neighbour queries
in $G$ in constant time. In addition, the algorithm receives the number $n$ of vertices of $G$.
We assume that the vertices of $G$ are numbered $1,2,\ldots,n$.
The oracle accepts queries of the form $(i,j)$, for
$i\leq n$, and $j\leq d$, to which it responds
with the $j$-th neighbour of 
$i$, or with $\bot$ if $i$ has less than $j$ neighbours.

The \emph{running time} of the
algorithm is defined as usual, \ie with respect to $n$. We assume
a uniform cost model, \ie\hspace{-4pt}, we assume that all basic arithmetic operations 
including random sampling can be performed in constant time, regardless of 
the size of the numbers involved.

\textbf{Distance.}
For two graphs $G$ and $H$, both with $n$ vertices,
$\dist(G, H)$ denotes the minimum number of edges that
have to be modified (\ie inserted or removed) in $G$ and $H$ to make  $G$ and $H$
isomorphic. 
For $\epsilon \in [0,1]$,
we say $G$ and $H$ are 
$\epsilon$-\emph{close} if $\dist(G,H) \leq \epsilon d n$.
If $G, H$ are not $\epsilon$-\emph{close}, then they are \emph{$\epsilon$-far}.
Note that in particular, $G$ and $H$ are $\epsilon$-far if their vertex numbers differ. 
A graph $G$ is \emph{$\epsilon$-close} to a property $\mathcal{P}$ if $G$
is $\epsilon$-close to some $H \in \mathcal{P}$. Otherwise, $G$ is
\emph{$\epsilon$-far} from~$\mathcal{P}$. Note that this generalises Definition~\ref{def:farFromHam}.
\begin{definition}[$\epsilon$-tester]
	Let $\mathcal{P} \subseteq \mathcal{C}_d$  be a property and $\epsilon\in(0,1]$. 
	An  $\epsilon$-\emph{tester for $\mathcal{P}$} is a probabilistic
	algorithm with oracle access to an input $G\in \mathcal{C}_d$ and auxiliary input
	$n:=\abs{V(G)}$. The algorithm does the following.
	\begin{enumerate}
		\item If $G \in \mathcal{P}$, then the $\epsilon$-tester accepts with probability  at least ${2}/{3}$.
		\item If $G$ is $\epsilon$-far from $\mathcal{P}$, then the $\epsilon$-tester rejects
		with probability at least ${2}/{3}$.
	\end{enumerate}
	An $\epsilon$-tester is called a one-sided error tester if it accepts every graph $G\in \mathcal{P}$ with probability~$1$.
\end{definition}

The \emph{query complexity} of an $\epsilon$-tester is the maximum number of oracle queries made with respect to $n$. Let $f:\mathbb{N}\rightarrow \mathbb{R}$ be a function.
A property $\mathcal{P}$ is \emph{testable} with (one-sided error and) query complexity $f(n)$, if for each $\epsilon\in (0,1]$ and each $n$, there is a (one-sided error) $\epsilon$-tester
for $\mathcal{P}\cap \{G\in\mathcal{C}_d: |V(G)|=n\}$ on inputs from $\{G\in\mathcal{C}_d: |V(G)|=n\}$ with query complexity $f(n)$. 

We now obtain the following result as a corollary of Theorem~\ref{thm:existenceOfLocallyButFarFromHamiltonianGraphs}.
\begin{corollary}\label{cor:main}
	Hamiltonicity is not testable with one-sided error and  query complexity $o(n)$ in the bounded-degree model.
\end{corollary}

\begin{proof}
	Pick $d$ as in Theorem~\ref{thm:existenceOfLocallyButFarFromHamiltonianGraphs} and let $\mathcal{P}\subseteq \mathcal{C}_d$ be the class of all Hamiltonian graphs of degree at most $d$.
	Towards a contradiction, assume that for every $\epsilon\in (0,1]$ and $n\in \mathbb{N}$ there is a one sided-error $\epsilon$-tester for  $\mathcal{P}\cap \{G\in\mathcal{C}_d : |V(G)|=n\}$ with query complexity $o(n)$. Let $\delta,\epsilon \in (0,1)$ be constants such that there is a sequence of $d$-bounded degree graphs $(G_N)_{N\in \mathbb{N}}$ of increasing order such that $G_N$ is $\delta$-locally Hamiltonian and $\epsilon$-far from being Hamiltonian for every $N\in \mathbb{N}$. Note that $\delta$ and $\epsilon$ exist by Theorem~\ref{thm:existenceOfLocallyButFarFromHamiltonianGraphs}. Let $\mathcal{T}$ be an $\epsilon$-tester for $\mathcal{P}$ with query complexity $f(n)\in o(n)$.	
	Since $f(n)\in o(n)$ there must be $n_0\in \mathbb{N}$ such that $f(n)\leq \delta n$ for all $n\geq n_0$. Let $N\in \mathbb{N}$ such that $|V(G_N)|\geq n_0$. Since $G_N$ is $\epsilon$-far from $\mathcal{P}$ there must be a sequence  of queries $(q_1,\dots, q_m)$ with $m\leq  \delta n$ such that $\mathcal{T}$ queries the sequence $(q_1,\dots, q_m)$ with non-zero probability and rejects $G_N$ with non-zero probability  after performing the queries $(q_1,\dots,q_m)$. Let $S$ be the set of vertices $v\in V(G_N)$ such that there is a query $q_i=(v,j)$ for $i\in [m]$. Because $G_N$ is $\delta$-locally Hamiltonian and $|S|\leq  \delta n$ there is a graph $H\in \mathcal{P}$ on $n$ vertices and $T\subseteq V(H)$ such that there is an isomorphism $G_N[N_{G_N}(S)]$ to $H[N_H(T)]$ which maps $S$ to $T$. Hence, after renaming the vertices in $N_H(T)$, the tester 	
	$\mathcal{T}$ gets exactly the same answers for queries in $q_1,\dots,q_m$ for $G_N$ and $H$. This implies that $\mathcal{T}$ queries the sequence $(q_1,\dots,q_m)$ in $H$ with  non-zero probability  and hence must reject $H$ with non-zero probability. This contradicts the assumption that $\mathcal{T}$ was a one-sided error tester for Hamiltonicity. 
\end{proof}	
\begin{note}
	Note that the above argument is not sufficient for two-sided error testers. This is the case because a two-sided error tester would be allowed to reject $H$ with probability $<1/3$. As long as for sufficiently many other query sequences the two-sided error tester accepts $H$, it  might still  accept $H$ with probability  at least $2/3$.  
\end{note}

\subsubsection*{Acknowledgement}
We thank the reviewers for helpful comments that improved the exposition in this paper.

\nocite{*}
\bibliographystyle{abbrvnat}
\bibliography{nonHamiltonianGraphs}
\label{sec:biblio}

\end{document}